\documentclass[12pt]{amsart}

\usepackage{amssymb}
\usepackage{color}

\usepackage[hmargin=2.5cm,vmargin=2.5cm]{geometry}

\usepackage{graphicx}

\newtheorem{theorem}{Theorem}[section]
\newtheorem{lemma}[theorem]{Lemma}
\newtheorem{corollary}[theorem]{Corollary}

\newtheorem{definition}[theorem]{Definition}

\theoremstyle{remark}
\newtheorem{remark}[theorem]{Remark}
\newtheorem{example}[theorem]{Example}

\newcommand{\R}{\mathbb{R}}
\providecommand{\C}{\mathbb{C}}
\renewcommand{\C}{\mathbb{C}}
\newcommand{\Z}{\mathbb{Z}}

\newcommand{\p}{\partial}

\newcommand{\spaceH}{\mathcal{X}}
\newcommand{\groupS}{\mathcal{S}}
\newcommand{\groupR}{\mathcal{R}}

\newcommand{\cc}[1]{\overline{#1}}
\newcommand{\abs}[1]{\vert #1\vert}

\newcommand{\at}[1]{\vert\sb{#1}}
\newcommand{\Vbar}{\cc{V}}

\renewcommand\kappa{\varkappa}
\renewcommand\epsilon{\varepsilon}

\DeclareMathOperator{\Ind}{Ind}
\DeclareMathOperator{\Span}{span}
\DeclareMathOperator{\loc}{loc}

\usepackage{bm}

\title{Symmetry and Dirac points in graphene spectrum}

\author{Gregory Berkolaiko}
\address{Department of Mathematics, Texas A\&M University, College
 Station, TX 77843-3368, USA}
\author{Andrew Comech}
\begin{document}

\begin{abstract}
 Existence and stability of Dirac points in the dispersion relation
 of operators periodic with respect to the hexagonal lattice is
 investigated for different sets of additional symmetries.  The
 following symmetries are considered: rotation by $2\pi/3$ and
 inversion, rotation by $2\pi/3$ and horizontal reflection, inversion
 or reflection with weakly broken rotation symmetry, and the case
 where no Dirac points arise: rotation by $2\pi/3$ and vertical
 reflection.  
 
 All proofs are based on symmetry considerations.  In particular,
 existence of degeneracies in the spectrum is deduced from the
 (co)representation of the relevant symmetry group.  The conical
 shape of the dispersion relation is obtained from its invariance
 under rotation by $2\pi/3$.  Persistence of conical points when the
 rotation symmetry is weakly broken is proved using a geometric phase
 in one case and parity of the eigenfunctions in the other.
\end{abstract}

\maketitle

%%%%%%%%%%%%%%%%%%%%%%%%%%%%%%%%%%%%%%%%%%%%%%%%%%%%%%%%%%%
\section{Introduction}
\label{sec:intro}

Many interesting physical properties of graphene
\cite{Nov_nob10,CasGraphen_rmp09,Katsnelson_graphene,FefWei_cmp14},
are consequences of presence of special conical points in the
dispersion relation, where its different sheets touch to form a
conical singularity.  These points are often referred to as Dirac
points or as diabilical points.

Most mathematical analyses of the dispersion relation of graphene are
performed in physics literature in the tight-binding approximation,
starting from the work of Wallace \cite{Wal_pr47} and Slonczewski and
Weiss \cite{SloWei_pr58}.  This is equivalent to modeling the material
as a discrete graph with vertices at the carbon molecules' locations
and with edges indicating chemical bonds.  A richer mathematical model
for graphene was considered by Kuchment and Post in
\cite{KucPos_cmp07}, who studied
honeycomb quantum graphs with even potential on edges.

The Schr\"odinger operator $H_\epsilon = -\Delta + \epsilon
q(\vec{x})$
in $\R^2$
with the real-valued potential $q(\vec{x})$ that has
honeycomb symmetry was considered by Grushin \cite{Gru_mn09}.  A
condition for a multiple eigenvalue to be a conical point was
established and checked in the perturbative regime of a weak potential
(small $\epsilon$).  The multiplicity two of the eigenvalue was proved
from the symmetry point of view, an approach that we fully develop
here.

The case of potential of arbitrary strength was studied by Fefferman
and Weinstein \cite{FefWei_jams12} (see also \cite{FefLeeWei_prep16}
for further results).  The results of \cite{FefWei_jams12} can be
schematically broken into three parts: (a) establish that the
dispersion relation has a double degeneracy at certain known values of
quasi-momenta; (b) establish that for almost all $\epsilon$ the
dispersion relation is conical in the vicinity of the degeneracy; (c)
prove that the conical singularities survive under weak perturbation
which destroys some of the symmetries of the potential (namely, the
rotational symmetry).  These results are contained in \cite[Thms
5.1(1), 4.1 and 9.1]{FefWei_jams12} with proofs which are rather
technical.

The purpose of this article is to make explicit the role of symmetry
in existence and stability of Dirac points and to give proofs that are
at the same time simpler and more general.  Our methods apply to many
different settings: graphs (discrete or quantum) and Schr\"odinger and
Dirac operators on $\R^2$.  We use Schr\"odinger operator as our
primary focus, and give numerical examples based on discrete graphs.  We
also consider the effect of different symmetries, substituting
inversion symmetry, usually assumed in the literature, with horizontal
reflection symmetry (the results are analogous or stronger, as
explained below).

We will now briefly review our results and the methods employed.  The
Schr\"odinger operator is assumed to be shift-invariant with respect
to the hexagonal lattice.  We also consider the following symmetries
(see Fig.~\ref{fig:lattice_symm} for an illustration): rotation by
$2\pi/3$ (henceforth, ``rotation''), inversion (reflection with
respect to the point $(0,0)$), horizontal reflection and, to a lesser
extent, vertical reflection.  We remark that horizontal and vertical
reflections are substantially different because the hexagonal lattice
is not invariant with respect to rotation by $\pi/2$.  We study the
question of existence and stability of Dirac points
when the operator has various subsets of the above symmetries.

We show that existence of the degeneracy is a direct consequence of
symmetries of the operator.  The natural tool for studying this is, of
course, the representation theory.  It is well known that existence of
a two- (or higher-) dimensional irreducible representation suggests
that some eigenvalues will be degenerate.  However, rotation combined
with inversion --- the most usual choice of symmetries
\cite{Gru_mn09,FefWei_jams12} --- is an abelian group, whose
irreps are all one-dimensional.  The resolution of this
question lies in the fact that the relevant symmetry is the inversion
combined with complex conjugation and one should look at
representations combining unitary and antiunitary operators, the
so-called \emph{corepresentations} introduced and fully classified by
Wigner \cite[Chap.~26]{Wigner_group_theory}.

To prove
the existence of
the degeneracy in the spectrum
(Lemma~\ref{lemma:mult}) we identify the 2-dimensional
(co)representation responsible for it and describe the subspace of the
Hilbert space that carries this representation.  We also relate our
results to the proofs of isospectrality, in particular the
isospectrality condition of Band--Parzanchevski--Ben-Shach
\cite{BanParBen_jpa09,ParBan_jga10}.

The conical nature of the dispersion relation is known to be generic
(see, for example, \cite[Appendix 10]{Arnold_mat_methods}); to prove
this in a particular case one uses perturbation theory, as done in
\cite{Gru_mn09} and, implicitly, in \cite{FefWei_jams12}.  Again, we
seek to make the effect of symmetry most explicit here.  This is done
on two levels.  First, in Lemma~\ref{lemma:sym_disp} and
Lemma~\ref{lem:symm_disp} we show that the dispersion relation also
has rotational symmetry and thus, by Hilbert-Weyl theory of invariant
functions, is restricted to be a circular cone (which could be
degenerate) plus higher order terms.  Then, in
Lemma~\ref{lem:intertwiner}, we show that the symmetries also enforce
certain relations on the first order terms of the perturbative
expansion of the operator, which restricts the possible form of the
terms.  In spirit, this conclusion parallels the Hilbert-Weyl theory,
but is more powerful: for example, it allows us to conclude that at
quasi-momentum $\vec{0}$, where we discover persistent degeneracies
with only the rotational symmetry, the dispersion relation is locally
flat.

Part (c) of the above classification, the survival of the Dirac points
when a weak perturbation breaks the rotational symmetry, can be
established by perturbation theory and implicit function theorem, as
done in \cite{FefWei_jams12}.  However, such resilience of
singularities indicates that there are topological obstacles to their
disappearance \cite{ManGuiVoz_prb07,Man_prb12,MonPan_jsp14}.  The
method familiar to physicists is to use the Berry phase
\cite{Ber_prsla84,Sim_prl83}, which works when the operator has
inversion symmetry (Section~\ref{sec:keepingV}).  Interestingly, when
instead of inversion symmetry we have horizontal reflection symmetry,
Berry phase is \emph{not} restricted to the integer multiples of $\pi$
and the topological obstacle has a different nature.  The survival of
the Dirac cone is shown to be a consequence of the structure of
representation of the reflection symmetry
(Section~\ref{sec:keepingF}), which combines eigenfunctions of
different parities at the degeneracy point.  As a consequence of our
proof we conclude that the perturbed cone, although shifted from the
corner of the Brillouin zone, remains on a certain explicitly defined
line.  In particular, this restricts the location of points in the
Brillouin zone where Dirac cones can be destroyed by merging with
their symmetric counterparts.  Naturally, this effect is also present
when there is horizontal reflection symmetry \emph{in addition} to the
inversion symmetry.  We remark that experimentally created potentials
usually possess the reflection symmetry,
\cite{BahPelSeg_ol08,Tar+_n12}.

In connection with the survival of the Dirac points, we would like to
mention the complementary result of by Colin de Verdi\`ere in
\cite{CdV_msmf91}, who considered the Schr\"odinger operator
$H_\epsilon = -\Delta + \epsilon q(\vec{x})$ with $q(\vec{x})$
periodic, real and inversion-symmetric, but \emph{not}
$2\pi/3$-rotation invariant.  In this case, for small $\epsilon$,
there are also conical singularities of the dispersion in the vicinity
of the same special quasi-momenta.  The proof uses the transversality
condition of von Neumann--Wigner \cite{vNeWig_pz29} and Arnold
\cite{Arn_fap72}.  The method of \cite{CdV_msmf91} or, on a more basic
level, the implicit function theorem, could also be used to prove our
results, but we feel that the Berry phase technique is both beautiful
and relatively unknown in the mathematics literature and thus deserves
an appearance.

To summarize, in addition to providing simpler and shorter
symmetry-based proofs to existing results, we discover some previously
unknown consequences.  In particular, we consider the case of
rotational symmetry coupled with horizontal reflection symmetry; in
this case, when the rotational symmetry is weakly destroyed, the
conical points travel on a special line.  We observe degeneracies at
quasi-momentum $\vec{0}$ in presence of rotational symmetry only; the
dispersion relation at this point is shown to be locally flat.
Finally, we explain why the coupling of rotation and vertical
reflection \emph{does not}, in general, lead to the appearance of
Dirac points.  The tools developed in this article would be easily
extensible to other lattice structures \cite{DoKuc_nsmmta13} and
graphene superlattices \cite{Yan+_np12,Pon+_n13}.

%%%%%%%%%%%%%%%%%%%%
\subsection{Symmetries}

The periodicity lattice of the operators that we consider is the
2-dimensional hexagonal lattice $\Gamma$ with the basis vectors
\begin{equation}
  \label{eq:lattice_vectors}
  \vec{a_1} = 
  \begin{pmatrix}
    \sqrt{3}/2 \\
    1/2
  \end{pmatrix},
  \qquad 
  \vec{a_2} = 
  \begin{pmatrix}
    \sqrt{3}/{2} \\ -1/2
  \end{pmatrix},
\end{equation}
see Fig.~\ref{fig:lattice_symm}(a).  The operator considered will
always be assumed to be invariant with respect to the shifts by this
lattice.

In addition to the shifts, the lattice $\Gamma$ has several other
symmetries.  We now describe some of them as operators acting on
functions on $\R^2$ (or on a graph embedded into $\R^2$).

\begin{figure}[t]
  \centering
  \includegraphics[scale=0.9]{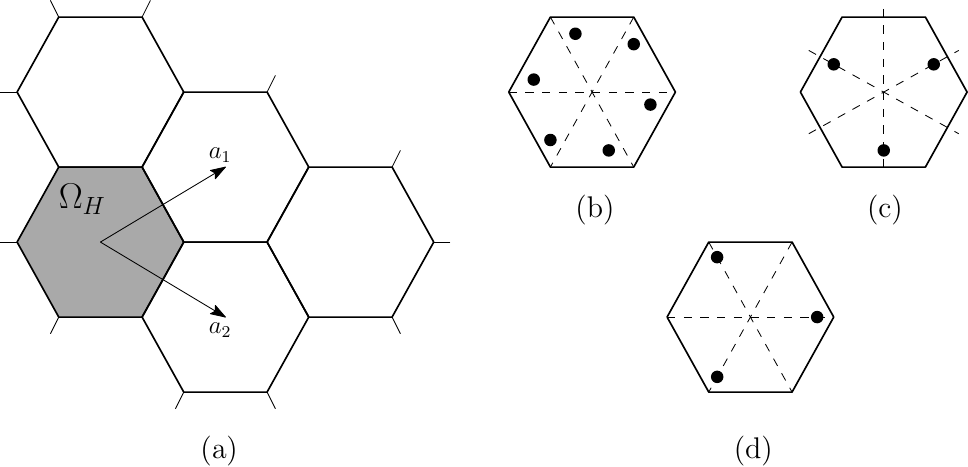}
  \caption{Hexagonal lattice (a) and examples of fundamental domains
    with symmetry $R$ and, additionally, (b) inversion symmetry $V$,
    (c) horizontal reflection symmetry $F$ and (d) vertical reflection
    symmetry $F_V$.  Note that we do not expect conical points in
    operators with symmetries $R$ and $F_V$, see
    section~\ref{sec:isotypic_RF_V}.}
  \label{fig:lattice_symm}
\end{figure}

\begin{itemize}
\item Rotation $R$ by $2\pi/3$ in the positive
  (counter-clockwise) direction:
  \begin{equation*}
    R: \psi(x_1, x_2) 
    \mapsto \psi\left( -\frac{1}{2}x_1 + \frac{\sqrt{3}}{2}x_2,
      - \frac{\sqrt{3}}{2}x_1 - \frac{1}{2}x_2\right) = \psi\left( M_R
    \vec{x} \right).
    \end{equation*}
\item Inversion $V$:
  \begin{equation*}
    V: \psi(x_1, x_2) \mapsto \psi(-x_1, -x_2).
  \end{equation*}
\item Horizontal reflection $F$:
  \begin{equation*}
    F : \psi(x_1, x_2) \mapsto \psi(-x_1, x_2).
  \end{equation*}
\end{itemize}

Note that $R$ and $V$ together form the abelian group of rotations by
multiples of $\pi/3$.

We denote by $C$ the
antiunitary operation of
taking complex conjugation (or ``time-reversal'' in physics
terminology),
  \begin{equation}\label{def-C}
    C : \psi(x_1, x_2) \mapsto \overline{\psi(x_1, x_2)}.
  \end{equation}
In what
follows, we will assume our operator has symmetries generated by a
subset of the following: complex conjugation $C$, rotation $R$, reflection $F$,
conjugate inversion $\Vbar=VC$.

As the base operator (i.e. before we apply Floquet-Bloch analysis) we
will always take an operator with real coefficients, thus it will be
symmetric with respect to complex conjugation.  As it turns out, an
important role is played by the product of inversion and complex
conjugation, known as the $\mathcal{PT}$ (parity-time) transformation:

%% \begin{remark}
%% \ac{?????}
%% We expect that our approach can be extended
%% from hamiltonian systems
%% to
%% systems with $\mathcal{PT}$ (parity-time) invariance,
%% that is, when the complex-valued potential is invariant
%% with respect to the following transformation:
%% %%\begin{itemize}
%% %%\item Symmetry $\Vbar$:
%%   \begin{equation*}
%%     \Vbar=VC: \psi(x_1, x_2) \mapsto \cc{\psi(-x_1, -x_2)}.
%%   \end{equation*}
%% %%\end{itemize}
%% Note that $\Vbar$ is not a $\C$-linear operator; it is, however, a
%% linear operator over reals.
%% \end{remark}

Finally, we will also consider the vertical reflection symmetry:
\begin{itemize}
\item Vertical reflection $F_V$:
  \begin{equation*}
    F_V : \psi(x_1, x_2) \mapsto \psi(x_1, -x_2).
  \end{equation*}
\end{itemize}
Its effect is not the same as that of the horizontal reflection $F$
because the two symmetries are aligned differently with respect to the
lattice $\Gamma$.  In fact, in contrast to $F$, the presence of $F_V$
(in addition to symmetry $R$) does not generally lead to the
appearance of conical points in the dispersion relation.  This
negative result is also important to understand; we explain it in
section~\ref{sec:isotypic_RF_V}.

In Fig.~\ref{fig:lattice_symm}(b-d) we show the fundamental domain of
the lattice with defects that have symmetry $R$ in addition to $V$,
$F$ or $F_V$, correspondingly.

%%%%%%%%%%%
\subsection{Operators}
\label{sec:operators}

As our primary motivational example we use the two-dimensional
Schr\"o\-dinger operator
\begin{equation}\label{def-H}
H = -\Delta + q(\vec{x}),
\qquad
\vec{x}\in\R^2,
\end{equation}
with the real-valued
%%L1loc
potential $q(\vec{x})$ assumed to be bounded and periodic with
respect to the lattice $\Gamma$.  For general properties of the
dispersion relation of such operators we refer the reader to
\cite{AvrSim_ap78,Kuchment_floquet,Kuc_bams16}.

To generate simple numerical examples we use discrete Schr\"odinger
operators with potentials crafted to break or retain some of
the symmetries listed above.  More precisely, denote by $G=(V,E)$ an infinite
graph embedded in $\R^2$, with vertex set $V$ and edge set $E$.  The
embedding is realized by the mapping $\loc : V \to \R^2$ which gives
the location in $\R^2$ of the given vertex.  A transformation $T :\R^2
\to \R^2$ preserves the graph structure if $u_1,u_2\in V$ implies
existence of $u_1', u_2' \in V$ such that $T\loc(u_j) = \loc(u_j')$
and $u_1', u_2'$ are connected by an edge if and only if $u_1, u_2$
are connected.

The graph is $\Gamma$-periodic if the graph structure is preserved by
the shifts defining the lattice.  A graph with space symmetry $S$ is
defined analogously.

The Schr\"odinger operator is defined on the functions from
$\ell^2(\C^V)$ by
\begin{equation}
  \label{eq:Schrod_discr}
  (H f)_v = \sum_{(v,u) \in E} m_{v,u} (f_v - f_u) + q_v f_v,
\end{equation}
where the sum is over all vertices $u$ adjacent to $v$, $m_{v,u}>0$
are weights associated to edges (often, they are taken inversely
proportional to edge length) and $q : V\to \R$ is the discrete
%%AC electric
potential.  In our examples, the graph structure will be
compatible with all symmetries of the lattice $\Gamma$, while $m$ and
$q$ will be breaking some of the point symmetries (however, they will
always be
periodic).  The simplest $\Gamma$-periodic graph is shown in
Fig.~\ref{fig:graphene_tba}(a).  This is the graph arising as the
tight-binding approximation of graphene.

Note that the discrete Schr\"odinger operator on graphs with more than
two atoms per unit cell is not a mere mathematical curiosity since it
arises in studying the twisted graphene and graphene in a periodic
potential (superlattice); see \cite{Lui+_prl11,Yan+_np12,Wal+_prb13}
and references therein.

\begin{remark}
We will always assume that $H$ is
time-reversal (TR) invariant,
or, in other words,
satisfies
\[
CHC=H,
\]
with $C$ the complex conjugation from \eqref{def-C}.
In particular, in the case
\eqref{def-H},
this means that $q$ is real-valued,
while in the case
\eqref{eq:Schrod_discr}
this means that both $q$ and $m$ are real-valued.
\end{remark}

%%%%%%%%%%%%
\subsection{Floquet-Bloch reduction}
\label{sec:floquet}

Floquet theory can be thought of as a version of Fourier expansion,
mapping the spectral problem on a non-compact manifold into a
continuous sum of spectral problems on a compact manifold.  The
compact spectral problems are parametrized by the representations of
the abelian group of periods (shifts).

Denote by $\spaceH(\vec{k})$, $\vec{k}=(k_1,k_2) \in
\mathbb{T}^2:=
[0,2\pi)^2$ the
space of Bloch functions, i.e.\ locally $L^2$ functions satisfying
\begin{equation}
  \label{eq:bloch_def}
  \psi(\vec{x} + n_1 \vec{a_1} + n_2 \vec{a_2}) = e^{i (n_1 k_1 + n_2 k_2)} \psi(x),
  \qquad n_1, n_2 \in \Z.
\end{equation}
For functions $\psi \in \spaceH(\vec{k})$ which also belong to the domain of
$H$ it can be immediately seen that
\begin{equation*}
  (H\psi)(x + n_1 \vec{a_1} + n_2 \vec{a_2}) = e^{i (n_1 k_1 + n_2 k_2)} H\psi(x),
\end{equation*}
i.e.\ the space $\spaceH(\vec{k})$ is invariant under $H$.  By
$H(\vec{k})$ we will denote the restriction of the operator $H$ to the
space $\spaceH(\vec{k})$.  Its domain is $\spaceH^2(\vec{k})$, the
dense subspace of $\spaceH(\vec{k})$ consisting of functions that
locally belong to $L^2$ together with their derivatives up to the
second order.

\begin{figure}[t]
  \centering
  \includegraphics[scale=0.9]{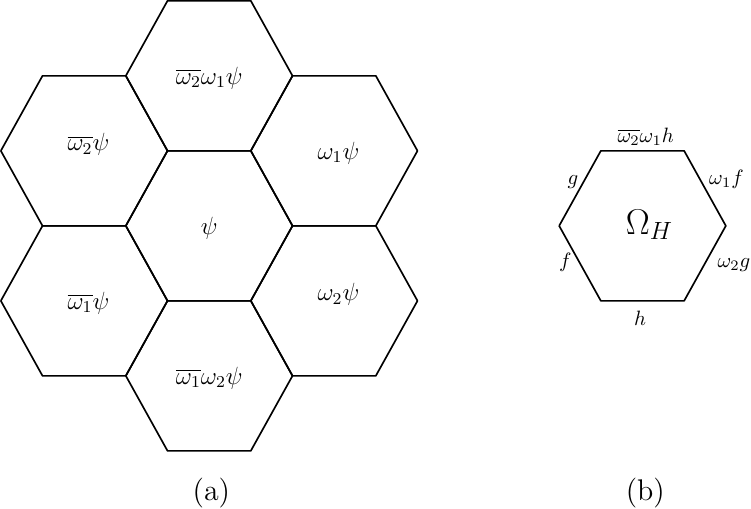}
  \caption{Floquet-Bloch reduction on the plane with hexagonal lattice
    generated by $\vec{a_1}$ and $\vec{a_2}$.}
  \label{fig:floquet_reduction}
\end{figure}

Choosing a fundamental domain\footnote{a domain having the property
  that each trajectory $\{\vec{x} + n_1 \vec{a_1} + n_2 \vec{a_2} : n_1, n_2 \in
  \Z\}$ has exactly one representative in it} of the action of the
group of periods, we can reduce the problem to the fundamental domain
with quasi-periodic boundary conditions.
The result of the Floquet-Bloch reduction is shown in
Fig.~\ref{fig:floquet_reduction}.  In Fig.~\ref{fig:lattice_symm}(a),
the lattice generating vectors $\vec{a_1}$ and $\vec{a_2}$ were shown
together with a convenient choice of the fundamental region (shaded)
and its four translations, by $\vec{a}_1$, $\vec{a}_2$, $\vec{a}_1 -
\vec{a}_2$ and $\vec{a}_1 + \vec{a}_2$.
We will denote this choice of the fundamental domain by $\Omega_H$.
The values of a
Bloch function in surrounding regions, according to
equation~\eqref{eq:bloch_def}, are indicated in
Fig.~\ref{fig:floquet_reduction}(a); we use the notation
\begin{equation}
  \label{eq:notatiok_k_omega}
  \omega_j = e^{i k_j},
\qquad
j=1,\,2.
\end{equation}
The continuity of the function and its derivative across
the boundaries of copies of the fundamental region impose boundary
conditions shown schematically in
Fig.~\ref{fig:floquet_reduction}(b).  They should be understood as
follows: taking the bottom and top boundaries as an example, and
parametrizing them left to right, the conditions read
\begin{equation*}
  \psi \big|_{\mathrm{top}} 
  = \overline{\omega_2}\omega_1 \psi \big|_{\mathrm{bottom}},
  \qquad
  -\partial_{\vec{n}} \psi \big|_{\mathrm{top}} 
  = \overline{\omega_2}\omega_1 \partial_{\vec{n}} 
  \psi \big|_{\mathrm{bottom}},
\end{equation*}
where the normal derivative is taken in the outward direction (this
causes the minus sign to appear).  We stress that in
Fig.~\ref{fig:floquet_reduction}(c) we use letters $f$, $g$ and $h$ as
placeholder labels, connecting the values of the function and its
derivative on similarly labeled sides.

\begin{figure}[t]
  \centering
  \includegraphics[scale=0.8]{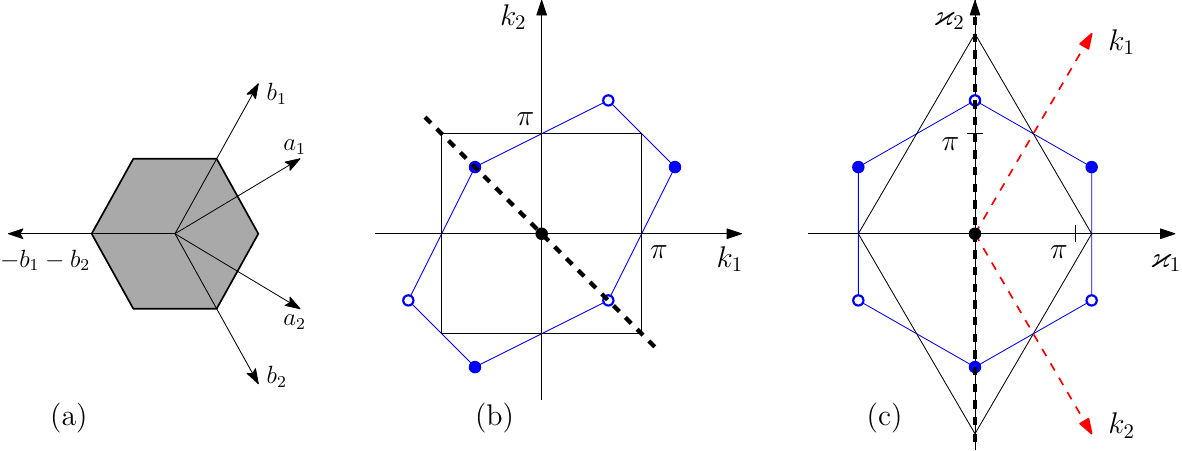}
  \caption{The dual basis (a) to the vectors $\vec{a_1}$ and
    $\vec{a_2}$ and two choices of the Brillouin zone in terms of (b)
    coordinates $k_1, k_2$ (drawn as if they were Cartesian) and (c)
    coordinates $\kappa_1, \kappa_2$ (which are Cartesian); part (c)
    also shows the correct position for the axes $k_1$ and $k_2$.  The
    axis of symmetry of the operator $\hat{F}$ is shown in dashed line
    (the equation $k_1=-k_2$).  Fixed points of the operator $\hat{R}$ are
    shown by circles (different fill styles correspond to different
    points of symmetry).}
  \label{fig:bri_zone}
\end{figure}

To represent the exponent of the Bloch phase $n_1k_1 + n_2k_2$ as a
scalar product, we introduce the vectors
\begin{equation}
  \label{eq:dual_lattice_vectors}
  \vec{b_1} = \left(\frac1{\sqrt{3}}, 1\right)^T, \qquad 
  \vec{b_2} = \left(\frac1{\sqrt{3}}, -1\right)^T,
\end{equation}
see Fig.~\ref{fig:bri_zone}(a).
Then
\begin{equation}
  \label{eq:dual_basis_def}
  \left.\vec{b}_i\right.^T\cdot \vec{a}_j = \delta_{i,j}.
\end{equation}
The vectors $\vec{b}_1$, $\vec{b}_2$ define a lattice which is known
as the \emph{dual lattice}.  For a hexagonal lattice, the dual lattice
is also hexagonal.  The lattice spanned by the vectors
$2\pi\vec{b}_1$, $2\pi\vec{b}_2$ will be denoted $\Gamma^*$.

Due to \eqref{eq:dual_basis_def},
one can write $n_1k_1 + n_2k_2$ as the dot product
\begin{equation*}
  n_1k_1 + n_2k_2 =
  \big(k_1\vec{b_1}+k_2\vec{b_2}\big)\cdot
  \big(n_1\vec{a_1}+n_2\vec{a_2}\big).
\end{equation*}

Let us comment on using coordinates $k_1, k_2$ which are the
coordinates with respect to the basis $\vec{b}_1, \vec{b}_2$ versus
the corresponding Cartesian coordinates $\kappa_1, \kappa_2$ given by
\begin{equation}
  \label{eq:kappa_def}
  \vec{\kappa} =
  \begin{pmatrix}
    \tfrac1{\sqrt{3}} & \tfrac1{\sqrt{3}}\\[3pt]
    1 & -1
  \end{pmatrix} \vec{k}
  =: B \vec{k} 
  \quad\mbox{or, equivalently,}\quad
  \vec{k} = B^{-1}\vec{\kappa} = 
  \begin{pmatrix}
    \tfrac{\sqrt{3}}2 & \tfrac12 \\[3pt]
    \tfrac{\sqrt{3}}2 & -\tfrac12
  \end{pmatrix} \vec{\kappa}.
\end{equation}
In Fig.~\ref{fig:bri_zone}(b) we show two
choices of the Brillouin zone\footnote{By ``Brillouin zone'' we
  understand \emph{any} choice of the fundamental domain of the dual
  lattice.  What is known as the ``first Brillouin zone'' is the
  hexagonal domain in blue in Fig.~\ref{fig:bri_zone}(c)} drawn in
terms of coordinates $k_1, k_2$ and coordinates $\kappa_1, \kappa_2$.
One arrives at the first picture if one uses $k_1$ and $k_2$ as parameters
for the dispersion relation (which is natural) ranging from $-\pi$ to
$\pi$ (black square) and then plots the result using $k_1$ and $k_2$
as Cartesian coordinates.  The resulting plot of the dispersion
relation will be skewed similarly to the blue hexagon in
Fig.~\ref{fig:bri_zone}(b) (cf. Figures 5 and 6 of
\cite{KucPos_cmp07}).  A more correct way of plotting is over a domain
in Fig.~\ref{fig:bri_zone}(c), as it will highlight the symmetries of
the result (see Figs.~\ref{fig:dispersion} and \ref{fig:contour} and
the explanations in the following section).

%%%%%%%%%%%%%%%%%%%%%%%%%%%%%%%%%%%%%%%%%%%%%%%%%%%
\section{Formulation of results}

For each value of the quasi-momentum $\vec{k}$, the operator $H(\vec{k})$
has discrete spectrum.  Its eigenvalues as functions of $\vec{k}$
%%AC changed 'is' to 'form' (better be plural)
form
what is known as the \emph{dispersion relation}.  Our results are
concerned with the structure of the dispersion relation for the
operators we described in Section~\ref{sec:operators}.  A typical
example is shown in Fig.~\ref{fig:dispersion}; it was computed for a
discrete Laplacian described in detail in
Example~\ref{exm:graph_structure}.

\begin{figure}[t]
  \centering
  \includegraphics[scale=0.6]{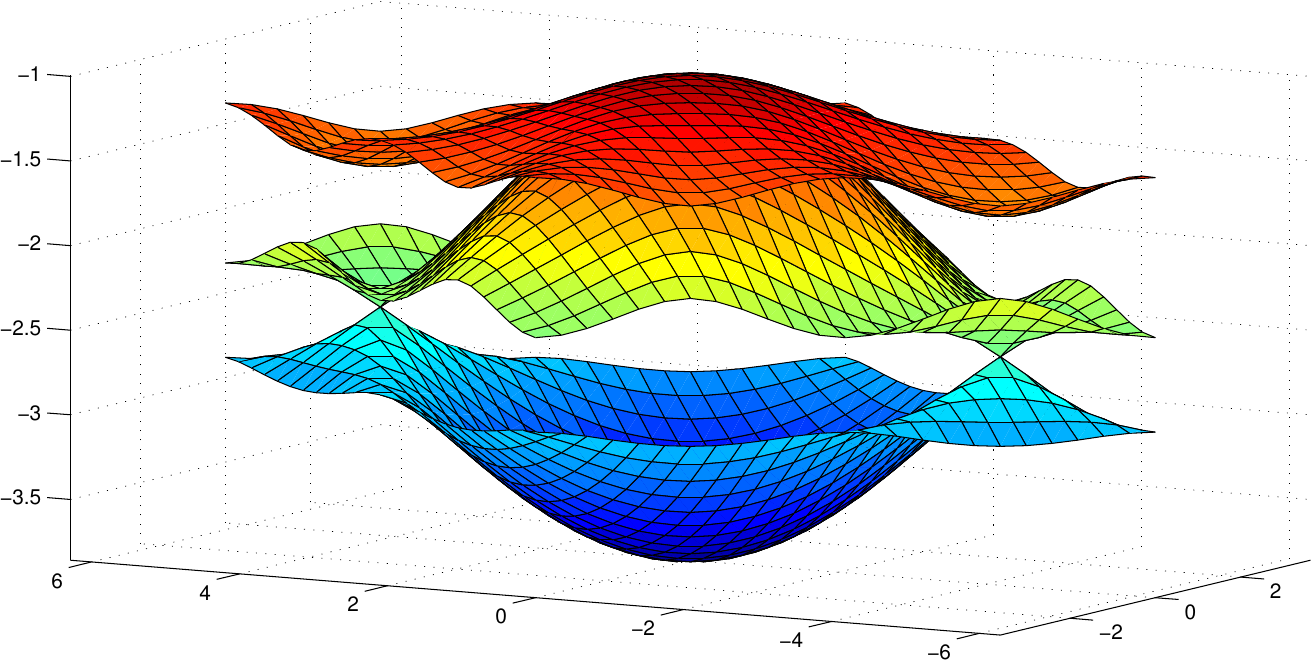}
  \caption{The lowest three bands of the dispersion relation of the
    graph from Example~\ref{exm:graph_structure}, which has reflection
    symmetry.  The lower two bands touch conically at the points
    $\pm\vec{k}^*$.  The Brillouin zone is parametrized by
    $\vec{\kappa}$ coordinates.}
  \label{fig:dispersion}
\end{figure}

In the figure, one can see two \emph{conical points} where the lowest
two sheets of the dispersion relation touch.  In terms of $\vec{k}$
coordinates, they touch at $\pm\vec{k}^*$, where
\begin{equation*}
  \vec{k}^* = \left(\frac{2\pi}{3}, -\frac{2\pi}{3}\right), 
  \quad\mbox{or, correspondingly,}\quad
  \vec{\kappa}^* = \left(0, \frac{4\pi}{3}\right).
\end{equation*}
The middle and the top sheets also touch, at the point $\vec{0} =
(0,0)$; at the point of contact both surfaces are locally flat.  We
will show that these features are typical: conical singularities at
the point $\vec{k}^*$ and flat contact at the point $\vec{0}$.

We start with formulating the following well-known result,
summarizing the effects the different symmetries of $H$ have
on the structure of the dispersion relation.

\begin{lemma}
  \label{lemma:sym_disp}
\begin{enumerate}
\item If the operator $H$ is $\Gamma$-periodic (i.e.\ invariant with
  respect to the shifts by the lattice $\Gamma$), then the dispersion
  relation $\lambda_n(\vec{\kappa})$ is $\Gamma^*$-periodic, i.e.\
  invariant with respect to the shifts
  \begin{equation}
    \label{eq:shift_action_dual_kappa}
    \vec{\kappa} \mapsto \vec{\kappa} + 2\pi b_1
    \qquad \mbox{and} \qquad
    \vec{\kappa} \mapsto \vec{\kappa} + 2\pi b_2.
  \end{equation}
\item
  If the operator $H$ is invariant with respect to
  complex conjugation $C$ or inversion $V$,
  then the dispersion relation $\lambda_n(\vec{\kappa})$ is
  invariant with respect to the inversion $\kappa\to-\kappa$.
\item
  If the operator $H$ is invariant
  with respect to horizontal reflection $F$,
  then the dispersion relation $\lambda_n(\vec{\kappa})$ is
  invariant with respect to the reflection
  $(\kappa_1,\kappa_2)\to(-\kappa_1,\kappa_2)$.
\item
  If the operator $H$ is invariant with respect to rotation $R$,
  then the dispersion relation $\lambda_n(\vec{\kappa})$ is
  invariant with respect to
  rotation by $2\pi/3$ around the point $\vec{0} = (0,0)$.  
  \begin{enumerate}
  \item If, in addition to symmetry $R$,
    the operator $H$ is $\Gamma$-periodic, then the dispersion
    relation is also invariant with respect to rotation by $2\pi/3$
    around the points $\pm \vec{\kappa}^* := \pm(0,4\pi/3)$.
  \item If, in addition to symmetry $R$,
    the operator $H$ has symmetry $V$ or $C$, the
    dispersion relation is invariant with respect to rotation by
    $\pi/3$ around the point $\vec{0} = (0,0)$.
  \end{enumerate}
\end{enumerate}
\end{lemma}

For completeness, we provide the proof in
Section~\ref{sec:reduced_symm}.

\begin{remark}
  When $H$ is invariant with respect to complex conjugation, inversion
  symmetry of the operator does not result in any additional
  symmetries of the dispersion relation.
\end{remark}

\begin{example}
  Figure~\ref{fig:dispersion} was produced for a $\Gamma$-periodic
  graph operator which has symmetries $R$, $C$ and $F$ (but not $V$).
  Its dispersion relation therefore has symmetry groups $D_6$ around
  the point $\vec{0}$, and $D_3$ around the points
  $\pm \vec{\kappa}^*$ ($D_3$ and $D_6$ are the groups of symmetries
  of equilateral triangle and hexagon).  This can be seen clearly if
  we plot the level curves of the dispersion surfaces,
  Fig.~\ref{fig:contour}.
\end{example}

\begin{figure}[t]
  \centering
  \includegraphics[scale=0.5]{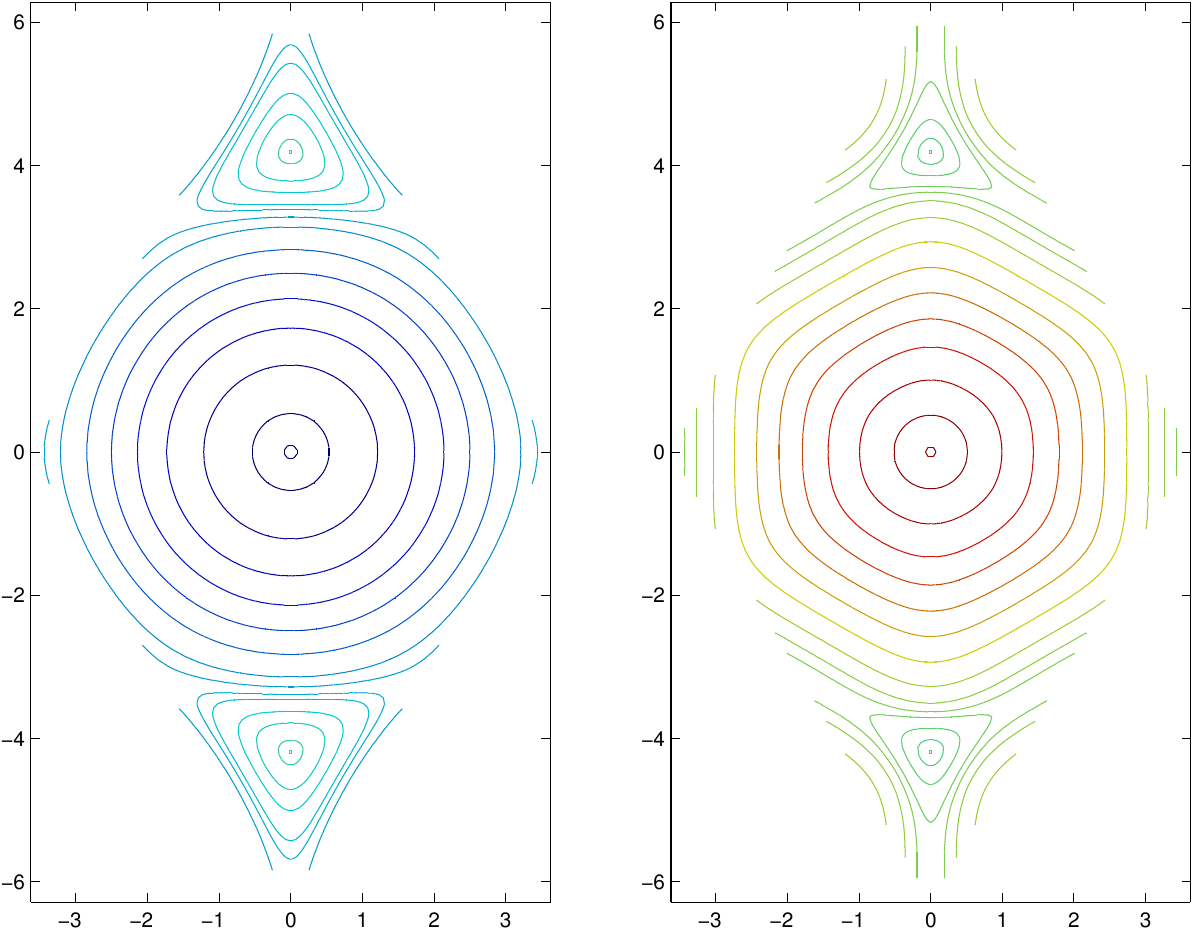}
  \caption{The contour plots of the two lowest bands from
    Fig.~\ref{fig:dispersion}.  Because of the symmetry of the
    operator (rotation, complex conjugation and reflection), the
    contours have the symmetry of an equilateral hexagon around
    $(0,0)$ and equilateral triangle around the
    points $\pm\vec{\kappa}^*$.}
  \label{fig:contour}
\end{figure}

\begin{theorem}
  \label{thm:main}
  Let the self-adjoint $\Gamma$-periodic operator $H$ be invariant
  under rotation $R$.  Let $\vec{k}_0$ be one of the points
  $\vec{k}^*$, $-\vec{k}^*$ or $\vec{0}$.  The space
  $\spaceH(\vec{k}_0)$ splits into the orthogonal sum
  \begin{equation}
    \label{eq:split_spaceH_star1}
    \spaceH(\vec{k}_0) = \spaceH_0(\vec{k}_0) \oplus
    \spaceH_\perp(\vec{k}_0),
  \end{equation}
  where
  $\spaceH_0(\vec{k}_0) = \{\psi \in \spaceH(\vec{k}_0) : R\psi =
  \psi\}$. This splitting is $H$-invariant.  Additionally,
  \begin{enumerate}
  \item \label{itm:main_kstar} 
    if $H$ is also invariant with respect to at least one of the
    following: reflection $F$ or the conjugated inversion $\Vbar$,
    then all eigenvalues of the operator $H$ restricted to
    $\spaceH_\perp(\pm\vec{k}^*)$ have even multiplicity.
Hence,
if $\spaceH_\perp(\pm\vec{k}^*)\ne 0$,
    $H(\pm\vec{k}^*)$ has some eigenvalues with multiplicity at least
    2.  If, moreover, the multiplicity of an eigenvalue $\lambda_0 \in
    \sigma(H(\pm\vec{k}^*))$ is exactly 2, the dispersion relation in
    coordinates $\vec\kappa$ is, to the leading order, a
    circular cone:
    \begin{equation}
      \label{eq:conical}
      (\lambda - \lambda_0)^2 = |\alpha|^2 |\vec{\kappa} - \vec{\kappa}_0|^2 +
      O\left(|\vec{\kappa} - \vec{\kappa}_0|^3\right),
      \qquad
      \alpha\in\C.
    \end{equation}
  \item If $H$ is also invariant under the complex conjugation $C$, then
    all eigenvalues of the operator $H$ restricted to
    $\spaceH_\perp(\pm\vec{0})$ have even multiplicity.
    Hence, if $\spaceH_\perp(\pm\vec{0})\ne 0$,
   $H(\pm\vec{0})$ has some eigenvalues with multiplicity at least
    2.  If, moreover, the multiplicity of an eigenvalue $\lambda_0 \in
    \sigma(H(\pm\vec{0}))$ is exactly 2, then the dispersion relation
    at this point is flat:
    \begin{equation}
      \label{eq:flat}
      (\lambda - \lambda_0)^2 = 
      O\left(|\vec{\kappa} - \vec{\kappa}_0|^3\right).
    \end{equation}
    % \ac{$V$ is not needed??  Can one use $V$ or $\Vbar$ instead of
    %   $C$??} Answers: no, no
  \end{enumerate}
\end{theorem}

%%AC
Theorem~\ref{thm:main} will follow from Lemma~\ref{lem:symm_disp},
Lemma~\ref{lemma:mult} (for the points $\vec{k}_0=\pm\vec{k}\sp\ast$)
and Lemma~\ref{lemma:point0} (for the point $\vec{k}_0=\vec{0}$).  In addition, 
in Lemma~\ref{lemma:structure_of_Hderiv} we will give a convenient expression 
for $\alpha$ of (\ref{eq:conical}).
We will also discuss a further splitting of the spaces $\spaceH_\perp(\vec{k}_0)$
and will give an explicit description of the restriction of
$H(\vec{k}_0)$ to the constituent subspaces.

By Theorem~\ref{thm:main}, we are guaranteed to have conical points
(i.e.\ points where the dispersion relation is of the form
(\ref{eq:conical})) whenever two conditions are satisfied: an
eigenvalue of $H$ on $\spaceH_\perp(\pm\vec{k}^*)$ has minimal
multiplicity (two) and is not in the spectrum of $H$ on
$\spaceH_0(\pm\vec{k}^*)$, and the parameter $\alpha \neq 0$.
Intuitively, it is clear that both conditions are ``generic'': if
either of them is broken, any typical small perturbation of the
potential should restore it.

To make this intuition precise, we consider the operator $H = -\Delta
+ \epsilon q(\vec{x})$, where we are able to say more about the
parameter $\alpha$ and the exact multiplicity of eigenvalues.

\begin{theorem}
  \label{thm:Laplacian}
  %%L1loc
  Let $H = -\Delta + \epsilon q(\vec{x})$ with bounded measurable real
  potential $q(\vec{x})$ which is invariant under the shifts by
  lattice $\Gamma$, rotation $R$, and at least one of the following:
  reflection $F$ or inversion $V$.
  Further, assume that the condition
  \begin{equation}
    \label{eq:cond_blah}
    \int_{\Omega_H} e^{\frac{4\pi i}{\sqrt3} x_1}
    q(\vec{x})\, d\vec{x} \neq 0
  \end{equation}
  is satisfied.
  Then the following conditions
hold for all $\epsilon\in\R$
except possibly on a discrete set:
  \begin{enumerate}
  \item there is an eigenvalue $\lambda_0(\epsilon)$ of $H$ on
    $\spaceH_\perp(\pm\vec{k}^*)$ of multiplicity exactly two and it is
    the smallest eigenvalue of $H$ on $\spaceH_\perp(\pm\vec{k}^*)$
    for small $\epsilon$,
  \item the eigenvalue $\lambda_0(\epsilon)$ is not an eigenvalue of
    $H$ on $\spaceH_0(\pm\vec{k}^*)$,
  \item the corresponding value of $\alpha$
    in equation (\ref{eq:conical}) is non-zero.
  \end{enumerate}
\end{theorem}

%%AC
Theorem~\ref{thm:Laplacian} will be proved in Section~\ref{sec:pert_pure}.
We mention that condition (\ref{eq:cond_blah}) above is equivalent to
condition (5.2) of \cite{FefWei_jams12} when one takes into account
symmetries (such as (2.36) of \cite{FefWei_jams12}).  

We now consider the fate of a conical point when the rotational
symmetry is broken by a small perturbation.
The following theorem is proved in Section~\ref{sec:persistence}.

\begin{theorem}
  \label{thm:persistence}
  Let $H$ be an operator satisfying the conditions of
  Theorem~\ref{thm:main}, part~\ref{itm:main_kstar}.  
  Assume that its dispersion relation has a
  nondegenerate conical point at the point $\vec{k}_0=\pm\vec{k}^*$.
  Consider the perturbed operator $H_\epsilon := H + \epsilon W$,
  where the relatively bounded perturbation $W$ has the same
  symmetries as $H$ (namely, $\Gamma$-invariance and either $\Vbar$- or
  $F$-invariance) \emph{except} the $R$-invariance.

  Then, for small $\epsilon$, the dispersion relation of $H_\epsilon$
  has a nondegenerate conical point in the neighborhood of
  $\vec{k}_0$.  Furthermore, if $H_\epsilon$ is invariant with respect
  to reflection $F$, the conical point remains on the line $k_2 =
  -k_1$ modulo $2\pi$.
\end{theorem}

We remark that a complementary result in the case when $H$ is the pure
Laplacian ($H=-\Delta$) and $W$ is a $\Gamma$ and $\Vbar$-invariant
(but \emph{not necessarily} $R$-invariant) potential satisfying a Fourier
condition akin to \eqref{eq:cond_blah} was obtained by Colin de
Verdi\`ere in \cite{CdV_msmf91}.  This highlights the fact that conical
singularities are very typical in 2-dimensional problems.

%%%%%%%%%%%%
\section{Symmetries in the dual space; proof of Lemma~\ref{lemma:sym_disp}}
\label{sec:reduced_symm}

We recall that
the operator $H(\vec{k})$ is the restriction of
the operator $H$ to the space $\spaceH(\vec{k})$.  Equivalently, it
can be considered as an operator on the compact domain of
Fig.~\ref{fig:floquet_reduction}(c) with the specified boundary
conditions\footnote{if the operator $H$ is specified on discrete
  graphs, the ``boundary conditions'' require special interpretation,
  see Section~\ref{sec:graph_examples} for some examples}.  It is
immediate from the definition of $H(\vec{k})$ that the dispersion
relation is invariant with respect to shifts by $2\pi$,
\begin{equation}
  \label{eq:shift_action_dual}
  \vec{k} \mapsto \vec{k} + (2\pi,0)
  \qquad \mbox{and} \qquad
  \vec{k} \mapsto \vec{k} + (0,2\pi).
\end{equation}
In other words, the dispersion relation is periodic with respect to
the lattice $\Gamma^*$.  We will now study other symmetries of the
dispersion relation.

For given values of $k_1, k_2$ (or, equivalently, $\omega_1,
\omega_2$, where $\omega_j = e^{ik_j}$), the operator $H(\vec{k})$ may
no longer have all the symmetries of the original operator $H$: while
the differential expression defining the operator is still invariant,
the domain of definition has been restricted and may not be invariant
anymore.

We start with the rotation operator $R$.  We first need to understand
the effect of $R$ on the space $\spaceH(\vec{k})$.  This can be
understood by rotating the picture in
Fig.~\ref{fig:floquet_reduction}(b) by $2\pi/3$ and finding the ``new
$\omega_1$, $\omega_2$'':
\begin{equation*}
  \omega_1' = \overline{\omega_1}\omega_2, \qquad
  \omega_2' = \overline{\omega_1}, \qquad
  \overline{\omega_2'}\omega_1' = \omega_2.
\end{equation*}
The last equation clearly follows from the first two.  For the
exponents $k'_1$, $k'_2$, defined as in \eqref{eq:notatiok_k_omega},
we have
\begin{equation}
  \label{eq:hatR_def}
  \begin{pmatrix}
    k'_1\\ k'_2
  \end{pmatrix}
  =
  \begin{pmatrix}
    -1 & 1\\
    -1 & 0 
  \end{pmatrix}
  \begin{pmatrix}
    k_1\\k_2
  \end{pmatrix}
  =: \hat{R}
  \begin{pmatrix}
    k_1\\k_2
  \end{pmatrix}.
\end{equation}
With respect to the dual basis $\vec{b_1}, \vec{b_2}$, the matrix
$\hat{R}$ is unitary: in terms of coordinates
$\vec{\kappa} = k_1 \vec{b_1} + k_2\vec{b_2} =: B\vec{k}$
the action of $R$ is given by
\begin{equation*}
  B \hat{R} B^{-1} = M_R^* = 
  \begin{pmatrix}
    -1/2 & -\sqrt{3}/2 \\
    \sqrt{3}/2 & -1/2
  \end{pmatrix}.
\end{equation*}
Therefore, the action of $\hat{R}$
is the rotation of coordinates by $2\pi/3$,
see Fig.~\ref{fig:bri_zone}(a), and $R$ acts
as a unitary operator from $\spaceH(\vec{k})$ to
$\spaceH(\hat{R}\vec{k})$.

More formally, denote by $S_{\vec{n}}$ the operator of the shift
$\psi(\vec x) \mapsto \psi(\vec x + n_1 \vec{a}_1 + n_2 \vec{a}_2) =:
\psi(\vec{x} + A\vec{n})$, with
\begin{equation}
  \label{eq:A_def}
  A := (\vec{a}_1, \vec{a}_2) = 
  \begin{pmatrix}
    \sqrt{3}/{2} & \sqrt{3}/{2} \\
    {1}/{2} & -{1}/{2}
  \end{pmatrix}.
\end{equation}
Then, for a function $\psi$ satisfying
\begin{equation*}
  \psi(\vec{x} + A \vec{n}) = e^{i\vec{k} \cdot \vec{n}} \psi(\vec{x}),
\end{equation*}
we have
\begin{align*}
  S_{\vec{n}} R \psi &= \psi\left(M_R(\vec{x} + A\vec{n})\right) 
  = \psi\left(M_R\vec{x} + A(A^{-1}M_RA)\vec{n}\right) \\
  &= e^{i \vec{k}\cdot(A^{-1}M_RA)\vec{n}} \psi(M_R\vec{x})
  = e^{i \left((A^{-1}M_RA)^*\vec{k}\right) \cdot \vec{n}} R\psi,
\end{align*}
and therefore $R$ maps functions from $\spaceH(\vec{k})$ to
$\spaceH(\hat{R}\vec{k})$ with $\hat{R} = (A^{-1}M_RA)^*$.

Since the operator $H(\vec{k})$ is the restriction of the operator
$H$, which is invariant under the rotation $R$, to the space
$\spaceH(\vec{k})$, we get 
\begin{equation}
  \label{eq:rot_operator}
  H(\vec{k}) =  R^* H( \hat{R} \vec{k} ) R,
\end{equation}
i.e.\ $H( \hat{R} \vec{k} )$ is unitarily equivalent to $H(\vec{k})$.
As a consequence, the dispersion
relation $\lambda_n(\vec{k})$ is invariant under the mapping
\begin{equation}
  \label{eq:rot_mod_map}
  \vec{k} \mapsto \hat{R} \vec{k} \mod 2\pi\mathbb{Z}^2,
\end{equation}
which maps a Brillouin zone to itself (here we assumed that $H$ is
$\Gamma$-periodic).  The fixed points of this mapping are the points
\begin{equation}
  \label{eq:fixedpoints}
  \vec{k}^* := (2\pi/3, -2\pi/3), 
  \qquad
  -\vec{k}^*:=(-2\pi/3, 2\pi/3),
  \qquad     
  \vec{0} := (0,0),
\end{equation}
and their shifts by $2\pi$.  In coordinates $\kappa$, the fixed points are
\begin{equation}
  \label{eq:fixedpoints_kappa}
  \vec{\kappa}^* := (0, 4\pi/3), 
  \qquad
  -\vec{\kappa}^*:=(0, -4\pi/3),
  \qquad 
  \vec{0} := (0,0).
\end{equation}

%%%
Analogous considerations for the horizontal reflection $F$ result in
\begin{equation*}
  \omega_1' = \cc{\omega_2}, \quad
  \omega_2' = \cc{\omega_1}, 
\end{equation*}
and, eventually, in
\begin{equation}
  \label{eq:F_action}
  F H(\vec{k}) F^* = H( \hat{F} \vec{k} ),
  \quad\mbox{where}\quad
  \hat{F} = 
  \begin{pmatrix}
    0 & -1\\
    -1 & 0
  \end{pmatrix}.
\end{equation}

The matrix $\hat{F}$ is a reflection with respect to the line
$k_2=-k_1$ and it leaves the points of this line invariant.  In
$\vec\kappa$ coordinates the mapping $\hat{F}$ acts as $(\kappa_1,
\kappa_2) \mapsto (-\kappa_1, \kappa_2)$.

%%%
Both complex conjugation and inversion result in
\begin{equation*}
  \omega_1' = \overline{\omega_1}, \qquad \omega_2' =
  \overline{\omega_2},
\end{equation*}
and possess a unique fixed point $\vec{k} = \vec{0}$.  However, their
composition $\Vbar$ preserves the space $\spaceH(\vec{k})$ for all values
of $\vec{k}$.  
To be more precise,
using the antiunitary operation of
taking complex conjugation $C$,
we have
\begin{equation}
  \label{eq:TV_action}
  C H(\vec{k}) C^{-1} = H( -\vec{k} ) = V H(\vec{k}) V^*.
\end{equation}

Equations (\ref{eq:rot_operator}), (\ref{eq:F_action}) and
(\ref{eq:TV_action}) show that the symmetries of the
operator result in the symmetries of the dispersion relation.
These symmetries have been summarized in
Lemma~\ref{lemma:sym_disp} above.

An important consequence of symmetry is a restriction on the possible
local form of the dispersion relation.  In particular,
the dispersion relation must be a circular cone
(which could be degenerate)
around a symmetry point of multiplicity two.

\begin{lemma}
  \label{lem:symm_disp}
  Let $\vec{\kappa}_0$ be one of the symmetry points, $\vec{0}$ or
  $\pm\vec{\kappa}^*$.  
  \begin{enumerate}
  \item \label{item:symm_disp_simple}
    If $\lambda_n(\vec{\kappa}_0) =: \lambda_0$ is a simple eigenvalue,
    the dispersion relation is given locally by
    \begin{equation}
      \label{eq:loc_disp_form_single}
      \lambda - \lambda_0 = a |\vec{\kappa} - \vec{\kappa}_0|^2 +
      %% \mbox{higher order terms}
      O(\abs{\vec{\kappa}-\vec{\kappa}\sb 0}^3), \qquad a \in \R.
    \end{equation}
  \item  \label{item:symm_disp_double}
    If $\lambda_n(\vec{\kappa}_0) =: \lambda_0$ is a double eigenvalue,
    the dispersion relation is given locally by
    \begin{equation}
      \label{eq:loc_disp_form_double}
      \lambda - \lambda_0 = \pm |\alpha| |\vec{\kappa} - \vec{\kappa}_0|
      + O(\abs{\vec{\kappa}-\vec{\kappa}\sb 0}^2),
\qquad
\alpha\in\C.
    \end{equation}
    Note that $\alpha$ may be equal to zero.
  \end{enumerate}
\end{lemma}

We note that using perturbation theory together with symmetry in
Sections~\ref{sec:appl_pert} and \ref{sec:point0} below, it will be
possible to make further conclusion about $\alpha$ appearing in
equation~\eqref{eq:loc_disp_form_double}..

\begin{proof}
  We start by remarking that by standard perturbation theory the
  number of eigenvalues close to $\lambda_0$ in the vicinity of the
  point $\vec{\kappa}_0$ remains equal to the multiplicity of
  $\lambda_0$ at $\vec{\kappa}_0$.
  
  %%L1loc
  We know from general theory of analytic Fredholm operators
  \cite{Zai+_rms75} that the dispersion relation is an analytic
  variety, i.e.\ given by an equation
  \begin{equation}\label{f-is-zero}
    F(\lambda, \vec{\kappa}) = 0,
  \end{equation}
  where $F$ is a real-analytic function.  Without loss of generality,
  consider the point $\vec{\kappa}_0=0$.  It is an easy special case
  of Hilbert-Weyl theorem on invariant functions
  \cite{Weyl_classicalgroups} (see also
  \cite[XII.4]{GolSch_groupsbif2}), that if a real-analytic function
  $f(\kappa_1,\kappa_2)$ is symmetric with respect to rotations by
  $2\pi/3$ around the origin, it can be represented as
  % \begin{equation}
  %   \label{eq:Hilbert_basis}
  $
  f(\kappa_1,\kappa_2) = g(\kappa_1^2+\kappa_2^2,\,
  \kappa_1^3-3\kappa_1\kappa_2^2,\, \kappa_2^3-3\kappa_2\kappa_1^2)
  $,
  % \end{equation}
  with some real-analytic $g$.
  % 
  % DERIVATION OF HILBERT BASIS.  DO NOT REMOVE!
  % 
  % Indeed represent $f$ as $f = \sum a_{\alpha\beta} z^\alpha
  % \cc{z}^\beta$ with $z = \kappa_1 + i \kappa_2$.  Since $f$ is a real
  % function, $\cc{a_{\alpha\beta}} = a_{\beta\alpha}$.  Since $f$ is
  % invariant with respect to rotation by $2\pi/3$, $f\left(e^{2\pi i/3}z\right)
  % = f(z)$, we have $a_{\alpha\beta} = a_{\alpha\beta} e^{2\pi
  %   i(\alpha-\beta)/3}$ and thus  $a_{\alpha\beta} = 0$ unless
  % $\alpha-\beta = 3k$, $k\in \mathbb{Z}$.  
  % 
  % Writing $f$ as
  % \begin{align}
  %   f &= \sum_{\alpha \geq \beta} a_{\alpha\beta} z^\alpha
  %   \cc{z}^\beta + c.c. \\
  %   &= \sum_{\alpha \geq \beta} (z\cc{z})^\beta
  %   \left[ b_{\alpha\beta}(z^{\alpha-\beta} + \cc{z}^{\alpha-\beta})
  %     +ic_{\alpha\beta}(z^{\alpha-\beta} - \cc{z}^{\alpha-\beta})\right].
  % \end{align}
  % Now we reduce by induction
  % \begin{align}
  %   z^{\alpha-\beta} \pm \cc{z}^{\alpha-\beta} 
  %   &= z^{3k} \pm  \cc{z}^{3k} \\
  %   &= (z^3+\cc{z}^3) \left(z^{3(k-1)} \pm  \cc{z}^{3(k-1)} \right)
  %   + (z\cc{z})^3 \left(z^{3(k-2)} \pm  \cc{z}^{3(k-2)} \right).
  % \end{align}
  Therefore,
  \eqref{f-is-zero}
  takes the form
  \begin{eqnarray}
    \label{g-is-zero}
    G(\lambda,\kappa_1^2+\kappa_2^2,
    \kappa_1^3-3\kappa_1\kappa_2^2,\, \kappa_2^3-3\kappa_2\kappa_1^2)=0,
  \end{eqnarray}
  with $G$ real-analytic in all the variables.
  
  If $\lambda=\lambda_0$ is a simple root,
  \[
  \p\sb\lambda G\at{(\lambda_0,0,0,0)}
  =\p\sb\lambda F\at{(\lambda_0,\vec{0})}
  \ne 0;
  \]
  by the implicit function theorem,
  \eqref{g-is-zero} defines
  $\lambda
  =\Lambda(\kappa_1^2+\kappa_2^2,
  \kappa_1^3-3\kappa_1\kappa_2^2,\, \kappa_2^3-3\kappa_2\kappa_1^2)$,
  with $\Lambda$ analytic in all three variables,
  and \eqref{eq:loc_disp_form_single} follows.
  
  If $\lambda=\lambda_0$ is a double root, we have
  \[
  \p\sb\lambda G\at{(\lambda_0,0,0,0)}
  =\p\sb\lambda F\at{(\lambda_0,\vec{0})}
  =0,
  \qquad
  \p\sb\lambda^2 G\at{(\lambda_0,0,0,0)}
  =\p\sb\lambda^2 F\at{(\lambda_0,\vec{0})}
  \ne 0.
  \]
  Without loss of generality, we assume that
  $\p\sb\lambda^2 G\at{(\lambda_0,0,0,0)}=2$.
  Then we have
  \begin{eqnarray}
    \label{eq:F_exp_double}
    &&
    F(\lambda, \vec{\kappa})
    =
    G(\lambda,\kappa_1^2+\kappa_2^2,
    \kappa_1^3-3\kappa_1\kappa_2^2,\, \kappa_2^3-3\kappa_2\kappa_1^2)
    \nonumber
    \\
    &&
    =
    (\lambda-\lambda_0)^2
    %%% - |\alpha|^2
    +a
    \left|\vec{\kappa}\right|^2 
    + O\left((\lambda-\lambda_0)^3\right) 
    + O\left((\lambda-\lambda_0)\left|\vec{\kappa}\right|^2\right)
    + O\left(\left|\vec{\kappa}\right|^3\right).
  \end{eqnarray}
  Note that the coefficient at $\left|\vec{\kappa}\right|^2$
  %% cannot be negative
  satisfies $a\le 0$ or else $F(\lambda,\vec{\kappa})$ would be
  strictly positive for $(\lambda,\vec{\kappa})$ close to
  $(0,\vec{0})$; thus, there would be no eigenvalues $\lambda$ for
  $\vec{\kappa}$ arbitrarily close to $\vec{\kappa}=\vec{0}$.
  
  If $a<0$,
  there is
  $\delta>0$ small enough
  and
  $K>0$ large enough so that
  for $\abs{\vec{\kappa}}<\delta$
  the function $F$ changes sign
  for $\lambda$ between $\lambda_0 +
  |a|^{1/2}\left|\vec{\kappa}\right| \pm
  K\left|\vec{\kappa}\right|^2$
  and also for $\lambda$ between
  $\lambda_0 - |a|^{1/2}\left|\vec{\kappa}\right| \pm
  K\left|\vec{\kappa}\right|^2$.
  Thus, the eigenvalue $\lambda$
  satisfies \eqref{eq:loc_disp_form_double}.

  If $a=0$, then
  we need higher order terms
  in the expansion of \eqref{eq:F_exp_double}:
  \begin{eqnarray}
    &&
    G(\lambda,\kappa_1^2+\kappa_2^2,
    \kappa_1^3-3\kappa_1\kappa_2^2,\, \kappa_2^3-3\kappa_2\kappa_1^2)
    \nonumber
    \\
    &&
    =(\lambda-\lambda_0)^2
    +c_0(\lambda-\lambda_0)\left|\vec{\kappa}\right|^2
    +c_1(\kappa_1^3-3\kappa_1\kappa_2^2)
    +c_2(\kappa_2^3-3\kappa_2\kappa_1^2)
    \nonumber
    \\
    &&
    \qquad
    +O\left((\lambda-\lambda_0)^2\left|\vec{\kappa}\right|^2\right)
    + O\left((\lambda-\lambda_0)^3\right) 
    + O\left(\left|\vec{\kappa}\right|^4\right)
    \nonumber
    \\
    &&
    =\mu^2
    +c_1(\kappa_1^3-3\kappa_1\kappa_2^2)
    +c_2(\kappa_2^3-3\kappa_2\kappa_1^2)
    + O\left(\mu^3\right)
    + O\left(\mu^2|\vec{\kappa}|^2\right)
    + O\left(\left|\vec{\kappa}\right|^4\right),
    \nonumber
  \end{eqnarray}
  where
  $\mu=\lambda-\lambda_0+c_0\left|\vec{\kappa}\right|^2/2$.
  We claim that
  $c_1=c_2=0$.
  For example, if we had $c_1>0$,
  then
  there would be $\delta>0$ such that
  $G$ is positive-definite
  for $|\mu|<\delta$,
  $\kappa_1\in(0,\delta)$, and $\kappa_2=0$;
  thus, there would be no eigenvalues $\lambda$
  for particular $\vec{\kappa}$ arbitrarily close
  to $\vec{\kappa}=\vec{0}$,
  leading to a contradiction.
  Once $c_1=c_2=0$,
  the relation
  $\mu^2+O(\mu^3)+O(\mu^2|\vec{\kappa}|^2)
  +O(|\vec{\kappa}|^4)=0$
  allows us to conclude that
  $\mu=O(|\vec{\kappa}|^2)$,
  which results in \eqref{eq:loc_disp_form_double} with $\alpha=0$.
\end{proof}

%%%%%%%%%%%%%%%%%%%%%%%%%%%%%%%%%%%%%%%%%%%%%%%%%%%%%%%%%%%%%%%
%%%%%%%%%%%%%%%%%%%%%%%%%%%%%%%%%%%%%%%%%%%%%%%%%%%%%%%%%%%%%%%
\section{Degeneracies in the spectrum at the point $\pm\vec{k}^*$}
\label{sec:mult}

We have seen in Section~\ref{sec:reduced_symm} that the points
$\vec{k} = \pm\vec{k}^*$ are special in that the operator
$H(\pm\vec{k}^*)$ has a large symmetry group.  In the next
subsection we give a review of the mechanism due to which symmetries give
rise to degeneracies in the spectrum.

%%%%%%%%%%%%%%%%%%%%%%%%%%%%%%%%%%%%%%%%%%%%%%%%%%%%%%%%%%%%%%%
\subsection{A review of representation theory background}
\label{sec:representations}

Let $H$ be a self-adjoint operator (``Hamiltonian'') acting on a
separable Hilbert space $\spaceH$.  Let $\groupS = \{\mathrm{Id}, S_1,
\ldots \}$ be a finite group of unitary operators on $\spaceH$ 
(the ``symmetries'' of $H$) which commute with $H$.

\begin{remark}
  It is assumed implicitly that the domain of $H$ is invariant
  under the action of operators $S\in\groupS$.  Such technical details
  will be omitted unless they have some importance to the task at hand.
\end{remark}

It is well-known (see,
e.g.~\cite{Wigner_group_theory,GoodmanWallach_symmetry}) that in the
circumstances described above, there is an \emph{isotypic
  decomposition} of $\spaceH$ into a finite orthogonal sum of
subspaces each carrying copies of an irreducible representation $\rho$ of
$\groupS$.  More precisely,
\begin{equation*}
  \spaceH = \bigoplus_{\rho} \spaceH_\rho,
\end{equation*}
where for any two vectors $v_1, v_2 \in \spaceH_\rho$, there is an
isomorphism between the spaces
\begin{equation*}
  [\groupS v_1] = \Span\left\{ Sv_1 : S\in\groupS \right\}
  \qquad \mbox{and} \qquad
  [\groupS v_2] = \Span\left\{ Sv_2 : S\in\groupS \right\},
\end{equation*}
which preserves the group action on the spaces (i.e.\ commutes with all
$S\in\groupS$). The dimension of $[\groupS v]$ is coincides with the
dimension of the representation $\rho$.

\begin{example}
  Let $\spaceH = L^2(\R)$ and $\groupS$ be the cyclic group of order 2
  generated by the reflection $x \mapsto -x$ or, more precisely,
  \begin{equation*}
    S: f(x) \mapsto f(-x).
  \end{equation*}
  Then $\spaceH = \spaceH_{\mathrm{even}} \oplus
  \spaceH_{\mathrm{odd}}$, where
  \begin{equation*}
    \spaceH_{\mathrm{even}} = \left\{ f\in \spaceH: f(-x) = f(x)
    \right\},
    \qquad
    \spaceH_{\mathrm{odd}} = \left\{ f\in \spaceH: f(-x) = -f(x)
    \right\}.
  \end{equation*}
  Then $\spaceH_{\mathrm{even}}$ carries infinitely many copies of the
  \emph{trivial} representation of $\groupS$:
  \begin{equation*}
    \mathrm{Id} \mapsto (1), \qquad S \mapsto (1),
  \end{equation*}
  while $\spaceH_{\mathrm{odd}}$ carries infinitely many copies of the
  \emph{alternating} representation of $\groupS$:
  \begin{equation*}
    \mathrm{Id} \mapsto (1), \qquad S \mapsto (-1).
  \end{equation*}
  Both representations are one-dimensional.  Note that the
  decomposition of a $\spaceH_\rho$ into irreducible copies is not
  unique.
\end{example}

Each isotypic component $\spaceH_\rho$ is invariant with respect to
$H$: either $Hv = 0$ or $H$ provides an isomorphism between subspaces
$[\groupS v]$ and $[\groupS Hv]$.

If $H$ has discrete spectrum then the restriction of $H$ to
$\spaceH_\rho$ has eigenvalues with multiplicities divisible by the
dimension of $\rho$.  Indeed, by commuting $\groupS$ and $H$ we see
that if $v$ is an eigenvector of $H$, then the entire subspace
$[\groupS v]$ is an eigenspace of $H$ with the same eigenvalue.

It is sometimes stated in the physics literature that if the group of
symmetries of an operator has an irreducible representation $\rho$,
the operator will have eigenspaces carrying this irreducible
representation; in particular, the corresponding eigenvalue will have
multiplicity equal to the dimension of $\rho$.
% This is not completely correct, as 
%%More accurately,
This implicitly assumes that
the isotypic component corresponding to this
representation is present in the domain of operator
(for examples to the contrary, see
e.g. \cite[Sec.~7.2]{BanParBen_jpa09}
or
Example~\ref{exm:graph_ex_0} below).
% provides another one) or artificial
%reasons (e.g. the operator may be restricted to another isotypic
%component).
Thus the fundamental question in describing spectral
degeneracies is finding the isotypic decomposition of the domain of
the operator.

%%%%%%%%%%%%%%%%%%%%%%%%%%%%%%%%%%%%%%%%%%%%%%%%%%%%%%%%%%%%%%%%%%%%%%
\subsubsection{$R$ and $F$ symmetry}
\label{sec:induced_RF}

Suppose the operator $H$ on the whole space has $R$ and $F$ symmetry.
The symmetries satisfy the relations $R^3 = F^2 = \mathrm{id}$ and
$FR^2 = RF$ and the symmetries group $\groupS$ is thus isomorphic to
the symmetric group $S_3$.  The representations are
\begin{align}
  \label{eq:RF_representation1triv}
  &R \mapsto (1), &&F \mapsto (1) &&\mbox{``trivial''},\\
  \label{eq:RF_representation1alt}
  &R \mapsto (1), &&F \mapsto (-1) &&\mbox{``alternating''},
\end{align}
and
\begin{equation}
  \label{eq:RF_representation2}
  R \mapsto
  \begin{pmatrix}
    \tau & 0\\
    0 & \overline{\tau}
  \end{pmatrix}, 
  \qquad F \mapsto
  \begin{pmatrix}
    0 & 1\\
    1 & 0
  \end{pmatrix}
 \qquad \mbox{"standard"},
\end{equation}
where $\tau$ is the third root of unity,
\begin{equation}
  \label{eq:tau_def}
  \tau := e^{2\pi i/3}.
\end{equation}

We thus expect that the two-dimensional representation will give rise
to eigenvalues of $H$ of multiplicity at least 2.

%%%%%%%%%%%%%%%%%%%%%%%%%%%%%%%%%%%%%%%%%%%%%%%%%%%%%%%%%%%%%%%%%%%%%
\subsubsection{$R$ and $V$ symmetry}
\label{sec:isotypic_RV}

On the face of it, the group generated by $R$ and $V$ is the group of
rotations by $\pi/3$, which is abelian and therefore has
one-dimensional representations only.  This would normally suggest
there are no persistent degeneracies in the spectrum.  However, the
symmetry relevant to us, as explained in
section~\ref{sec:reduced_symm}, is $V$ combined with complex
conjugation.  The representation $\rho(\Vbar)$ must be
an antiunitary operator, i.e.\ an operator $A$ satisfying
\begin{equation}
  \label{eq:def_antiunitary}
  A(\alpha v) = \cc{\alpha} (Av),
  \qquad \langle Av, Au\rangle = \langle u,v\rangle,
\end{equation}
which is a complex conjugation followed by the
multiplication by a unitary matrix.  Representations combining unitary
and antiunitary operators have been fully classified by Wigner
\cite[Chap.~26]{Wigner_group_theory} (see also \cite{BraDav_rmp68} for
a summary of the method), who called them ``corepresentations''.  In
short, one looks at the representation of the maximal unitary subgroup
(in our case, the cyclic group of rotations $R$) and, from them,
follows a simple prescription to construct all corepresentations.
This prescription is essentially constructing the induced
representation \emph{\`a la Frobenius}, although in the case when the induced representation decomposes into
two copies of an irrep, one takes only one copy.  

The group $\groupS$ has two corepresentations, given by
\begin{align}
  \label{eq:RV_representations1}
  &R : z \mapsto z,&  &\Vbar : z \mapsto \overline{z},\\
  \label{eq:RV_representations2}
  &R : 
  \begin{pmatrix}
    z_1 \\ z_2
  \end{pmatrix}
  \mapsto
  \begin{pmatrix}
    \tau z_1 \\ \overline{\tau} z_2
  \end{pmatrix},&  
  &\Vbar : 
  \begin{pmatrix}
    z_1 \\ z_2
  \end{pmatrix}
  \mapsto
  \begin{pmatrix}
    \overline{z_2} \\ \overline{z_1}
  \end{pmatrix}.
\end{align}

To see how they arise, we start with the representation
$\rho_1:\ R \mapsto (\tau)$ of the subgroup
$\groupR = \{\mbox{id}, R, R^2\}$, acting on a 1-dimensional space
spanned by $\vec{v}_1$.  We denote $\vec{v}_2 = \Vbar\vec{v}_1$ and
calculate
\begin{align}
  \label{eq:inducedV}
  & R \vec{v}_1 = \tau \vec{v}_1, 
  && \Vbar \vec{v}_1 = \vec{v}_2\\
  & R \vec{v}_2 = R\Vbar \vec{v}_1 = \Vbar R \vec{v}_1 
  = \Vbar \tau \vec{v}_1 = \cc{\tau} \Vbar \vec{v}_1 
  = \cc{\tau} \vec{v}_2, 
  && \Vbar \vec{v}_2 = \Vbar^2 \vec{v}_1 = \vec{v}_1.
\end{align}
This is the representation (\ref{eq:RV_representations2}) shown above.

The induced representation of $\rho_2: \ R \mapsto (\tau^2)$ is
the same, after the change of basis
$\vec{v_1} \leftrightarrow \vec{v}_2$.

The induced representation of the trivial representation $\rho_0: \ R
\mapsto (1)$ of $\groupR$ turns out to be
\begin{align}
  \label{eq:inducedV0}
  & R \vec{v}_1 = \vec{v}_1, 
  && \Vbar \vec{v}_1 = \vec{v}_2\\
  & R \vec{v}_2 = R\Vbar \vec{v}_1 = \Vbar R \vec{v}_1 
  = \Vbar \vec{v}_1 = \vec{v}_2, 
  && \Vbar \vec{v}_2 = \Vbar^2 \vec{v}_1 = \vec{v}_1.
\end{align}
After the change of basis $\vec{u}_1 = \vec{v}_1 + 
\vec{v}_2$, $\vec{u}_2 = i(\vec{v}_1 - \vec{v}_2)$, this
representation factorizes into two copies of representation
(\ref{eq:RV_representations1}) above.

Since we considered every representation of the subgroup $\groupR$,
this exhausts the list of corepresentations of $\groupS$.  We remark
that the bars over $z$ appear in
(\ref{eq:RV_representations1})-(\ref{eq:RV_representations2}) since
$z$ are scalar coefficients in the expansion over
$\{\vec{v}_1,\vec{v}_2\}$ and $\Vbar$ is antilinear,
equation~(\ref{eq:def_antiunitary}).

%%%%%%%%%%%%%%%%%%%%%%%%%%%%%%%%%%%%%%%%%%%%%%%%%%%%%%%%%%%%%%%%%%%%%
\subsubsection{$R$ and $C$ symmetry}
\label{sec:isotypic_RC}

As seen in Section~\ref{sec:reduced_symm}, at the point $\vec{k} =
\vec{0}$ the operator $H(\vec{k})$ will retain the symmetry with
respect to rotation $R$ and complex conjugation $C$.  So it is
important to consider the corresponding corepresentations.

Both the derivation and the answer are identical to the case of group
generated by $R$ and $\Vbar$: the symmetry group has two corepresentations, given by
\begin{align}
  \label{eq:RC_representations1}
  &R : z \mapsto z,&  &C : z \mapsto \overline{z},\\
  \label{eq:RC_representations2}
  &R : 
  \begin{pmatrix}
    z_1 \\ z_2
  \end{pmatrix}
  \mapsto
  \begin{pmatrix}
    \tau z_1 \\ \overline{\tau} z_2
  \end{pmatrix},&  
  &C : 
  \begin{pmatrix}
    z_1 \\ z_2
  \end{pmatrix}
  \mapsto
  \begin{pmatrix}
    \overline{z_2} \\ \overline{z_1}
  \end{pmatrix}.
\end{align}

%%%%%%%%%%%%%%%%%%%%%%%%%%%%%%%%%%%%%%%%%%%%%%%%%%%%%%%%%%%%%%%%%%
\subsubsection{$R$ and $F_V$ symmetry}
\label{sec:isotypic_RF_V}

Finally, we investigate what happens if the operator is symmetric with
respect to rotation $R$ and vertical reflection $F_V$.  The dual
action of $F_V$ is $(\kappa_1,\kappa_2)\mapsto(\kappa_1,-\kappa_2)$.
To preserve the fixed points $\pm\vec{\kappa}^*$, we need to pair
$F_V$ with $C$, i.e.\ consider the group generated by $R$ and
$\cc{F_V}$.  This group is $S_3$, yet we should be looking at
corepresentations, of which there are three, all one-dimensional,
\begin{align}
  \label{eq:RFV_rep0}
  &R : z \mapsto z,&  &\cc{F_V} : z \mapsto \overline{z},\\
  \label{eq:RFV_rep1}
  &R : z \mapsto \tau z,&  &\cc{F_V} : z \mapsto \overline{z},\\
  \label{eq:RFV_rep2}
  &R : z \mapsto \cc{\tau} z,&  &\cc{F_V} : z \mapsto \overline{z}.
\end{align}

This suggests that a typical problem\footnote{i.e.\ one without
  ``accidental'' degeneracies; it must be mentioned that the
  physically intuitive claim that ``accidental'' degeneracies do not
  happen generically remains, to a large extent, mathematically
  unproven; the best result in this direction is by Zelditch
  \cite{Zel_aif90}.} with these symmetries is not expected to have any
conical points in its dispersion relation.  According to
Lemma~\ref{lemma:point0}, there will still be generic degeneracies at
the point $\vec{0}$ but those are not conical.

% It is also easy to check that the induced representation of $\rho_1$
% decomposes into two copies of irrep (\ref{eq:RFV_rep1}) and the
% induced representation of $\rho_2$ into two copies of
% (\ref{eq:RFV_rep2}).

%%%%%%%%%%%%%%%%%%%%%%%%%%%%%%%%%%%%%%%%%%%%%%%%%%%%%%%%%%%%%%%
\subsection{Degeneracies in the spectrum of $H(\vec{k}^*)$}

The presence of degeneracies in the spectrum of the operator
$H(\vec{k})$ at the points $\pm\vec{k}^*$, which forms a part of
Theorem~\ref{thm:main}, follows directly from the representation
theory.

\begin{lemma}
  \label{lemma:mult}
  Let the self-adjoint operator $H$ be $\Gamma$-periodic and invariant
  under rotation $R$.  The space $\spaceH(\vec{k}^*)$, where
  $\vec{k}^* := \left(2\pi/3, -2\pi/3 \right)$, splits into the
  orthogonal sum
  \begin{equation}
    \label{eq:split_spaceH_star2}
    \spaceH(\vec{k}^*) = \spaceH_0(\vec{k}^*) \oplus
    \spaceH_\perp(\vec{k}^*),
  \end{equation}
  where
  $\spaceH_0(\vec{k}^*) = \{\psi \in \spaceH(\vec{k}^*) : R\psi =
  \psi\}$. This splitting is $H$-invariant.  

  If $H$ is also invariant with respect to at least one of the
  following: reflection $F$ or the conjugated inversion $\Vbar$, then
  all eigenvalues of the operator $H$ restricted to
  $\spaceH_\perp(\vec{k}^*)$ have even multiplicity.  Moreover, each
  eigenspace has an orthonormal basis $\{f^1_n, f^2_n\}$, such that
  \begin{equation}
    \label{eq:spec_basis}
    R f^1_n = \tau f^1_n, \qquad 
    R f^2_n = \cc{\tau}f^2_n, \qquad
    \mbox{and} \qquad
    f^2_n = F f^1_n \quad\mbox{or}\quad f^2_n = \Vbar f^1_n,     
  \end{equation}
  correspondingly.
\end{lemma}

\begin{proof}
  Since $H$ commutes with $R$, the space $\spaceH_0(\vec{k}^*)$ is
  $H$-invariant and, by self-adjointness, so is its orthogonal
  complement $\spaceH_\perp(\vec{k}^*)$.

  If $H$ is also invariant with respect to $\Vbar$, the isotypic
  component corresponding to representation
  (\ref{eq:RV_representations1}) is characterised by $R\vec{v} =
  \vec{v}$ and therefore coincides with $\spaceH_0(\vec{k}^*)$.  Thus
  the space $\spaceH_\perp(\vec{k}^*)$ is the isotypic component of
  representation (\ref{eq:RV_representations2}) and every eigenvalue
  of $H$ on this space is evenly degenerate.  Moreover, each
  eigenspace of dimension $2N$ has an orthonormal basis $\{f^1_n,
  f^2_n\}_{n=1}^N$, such that every pair $f^1_n$ and $f^2_n$ forms a
  basis of representation (\ref{eq:RV_representations2}).
  Namely, for all $z_1, z_2 \in \C$,
  \begin{equation}
    \label{eq:rep_realization}
    R\left(z_1 f^1_n + z_2f^2_n\right) 
    = \tau z_1 f^1_n + \cc{\tau} z_2 f^2_n 
    \qquad\mbox{and}\qquad
    \Vbar\left(z_1 f^1_n + z_2f^2_n\right) 
    = \cc{z_2} f^1_n + \cc{z_1}f^2_n    
  \end{equation}
  whence \eqref{eq:spec_basis} follows.

  By a similar reasoning, if the operator $H$ is $F$-invariant, the
  sum of isotypic components of (\ref{eq:RF_representation1triv}) and
  (\ref{eq:RF_representation1alt}) is characterised by $R\vec{v} =
  \vec{v}$ and therefore coincides with $\spaceH_0(\vec{k}^*)$.
  Again, the space $\spaceH_\perp(\vec{k}^*)$ is the isotypic
  component of the two-dimensional representation
  (\ref{eq:RF_representation2}) and the same conclusion follows.
\end{proof}

%%%%%%%%%%%%%%%%%%%%%%%%%%%%%%%%%%%%%%%%%%%%%%%%%%%%%%%%%%%%%%%%%%
\subsection{Explicit splitting of $H(\vec{k}^*)$ and connection to isospectrality}
\label{sec:full_split}

For computation, as well as for better understanding, it is
instructive to split the operator $H(\vec{k}^*)$ further.  It is easy
to show that the space $\spaceH_\perp(\vec{k}^*)$ splits further as
\begin{equation}
  \label{eq:12split}
  \spaceH_\perp(\vec{k}^*) 
  = \spaceH_1(\vec{k}^*) \oplus \spaceH_2(\vec{k}^*) 
  := 
  \{\psi \in \spaceH(\vec{k}^*) : R\psi = \tau \psi\} \oplus
  \{\psi \in \spaceH(\vec{k}^*) : R\psi = \tau^2 \psi\}.
\end{equation}
It is clear that the spaces $\spaceH_j(\vec{k}^*)$, $j=0,1,2$ are the
isotypic components of the full space with respect to the irreducible
representation $\rho_j$, $j=0,1,2$ of the symmetry subgroup of
rotations $\groupR = \{\mbox{id}, R, R^2\}$.  If $H$ is $R$-invariant,
it preserves the spaces $\spaceH_j(\vec{k}^*)$, $j=0,1,2$.

Moreover, the spaces $\spaceH_1(\vec{k}^*)$ and $\spaceH_2(\vec{k}^*)$
are mapped isomorphically to each other by $F$ or by $\Vbar$.  Thus,
if $H$ has appropriate symmetry, the restrictions of $H$ to these
spaces are unitarily equivalent and therefore isospectral.  The double
degeneracy of the spectrum of $H$ on $\spaceH_\perp(\vec{k}^*)$ is a
direct consequence of this fact.

\begin{figure}
  \centering
  \includegraphics[scale=0.75]{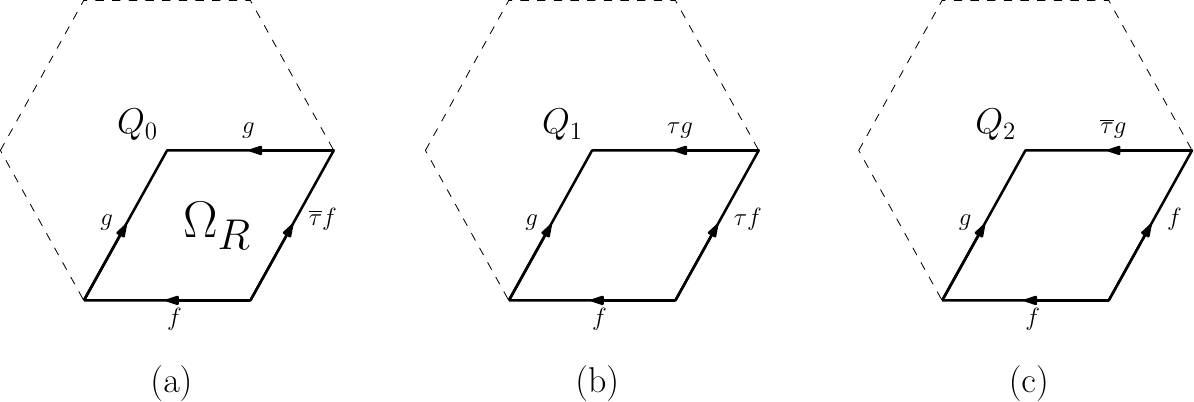}
  \caption{Operators $Q_0$, $Q_1$ and $Q_2$. We
      use the notation $\tau = e^{2\pi i/3}$.}
  \label{fig:operatorsQ}
\end{figure}

We can give an explicit description of the restrictions of $H$ to
$\spaceH_j(\vec{k}^*)$, $j=0,1,2$.  They are unitarily equivalent to
the differential operators $Q_j$ defined as follows.  Consider the
rhombic subdomain $\Omega_R$ covering $1/3$ of the hexagonal
fundamental domain, shown in Fig.~\ref{fig:operatorsQ}.  Denote by
$Q_j$, $j=0,1,2$, the operators having the same differential
expression as $H$ (see, for example, \eqref{def-H}) and with the
boundary conditions specified in Fig.~\ref{fig:operatorsQ}(a), (b) and
(c), correspondingly.  The equivalence of $Q_j$ to $H$ on the space
$\spaceH_j(\vec{k}^*)$ is realized by embedding the functions from
$L^2(\Omega_R)$ into $L^2(\Omega_H)$ by extending them by 0 and using
the operator
\begin{equation}
  \label{eq:T_def}
  T_j = \frac1{\sqrt{3}}\left(I + \tau^j R + \tau^{2j} R^2\right).
\end{equation}

The operators $Q_1$ and $Q_2$ are isospectral, as explained above.
The isospectrality can also be proved by a simple ``transplantation''
argument, similar to the proofs of isospectrality of certain domains
(such as the proof by Buser et al. \cite{Bus+_imrn94} for the
Gordon--Webb--Wolpert pair \cite{GorWebWol_im92}).  It can also be
checked using an algebraic condition of Band, Parzanchevski and
Ben-Shach, see \cite[Cor. 4.4]{BanParBen_jpa09} or
\cite[Cor. 4]{ParBan_jga10}.  Namely, if $\groupS$ is a symmetry group
of the operator $A$ and $H_1$, $H_2$ are subgroups of $\groupS$ with
the corresponding representations $\rho_1$ and $\rho_2$ such that the
induced representations
\begin{equation}
  \label{eq:bpb_condition}
  \Ind_{H_1}^\groupS \rho_1 \simeq \Ind_{H_2}^\groupS \rho_2
\end{equation}
are isomorphic, then the restrictions of $A$ to the isotypic
components of $\rho_1$ and $\rho_2$ are isospectral.
In our case, $H_1 = H_2 = \groupR$, the rotation subgroup, and the
representations $\rho_1$ are $\rho_2$ act by multiplication by
$\tau$ and $\tau^2$, respectively, with the induced representations
being precisely the two-dimensional representations
(\ref{eq:RF_representation2}) and (\ref{eq:RV_representations2}).

From the explicit description of the degenerate eigenstates of
$H(\vec{k}^*)$ as eigenvectors of $Q_1$ and $Q_2$, we get the
following practical corollary.

\begin{corollary}
  \label{cor:suppres}
  For any potential, the degenerate eigenstates of $H(\vec{k}^*)$
  vanish (\emph{are suppressed})
  at the center of the hexagonal fundamental domain.
\end{corollary}

\begin{proof}
  At the top left corner of the rhombic subdomain,
  Fig.~\ref{fig:operatorsQ}(a), the boundary conditions require
  $g=\tau g$.  This point is fixed by either the reflection or the
  inversion, thus both eigenfunctions have a zero there.
\end{proof}

%%%%%%%%%%%%
\subsection{Graph examples}
\label{sec:graph_examples}

While the splitting of Section~\ref{sec:full_split} was formulated for
continuous differential operators in $\R^2$, the method applies to
other models, such as graphs, with a little adjustment.  Here we
explain, by example, the construction of the operators $Q_j$.

\begin{figure}[t]
  \centering
  \includegraphics[scale=0.7]{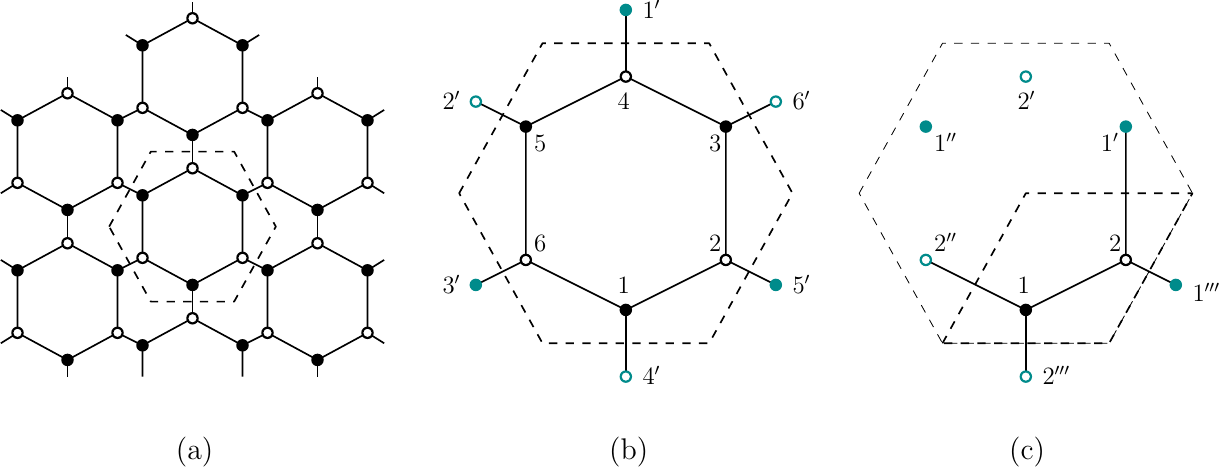}
  \caption{A discrete graph with symmetries $R$ and $F$.}
  \label{fig:graph_example}
\end{figure}

\begin{example}
  \label{exm:graph_structure}
  It is easier to start with an example that has a richer
  structure, such as the periodic graph of
  Fig.~\ref{fig:graph_example}(a).  It is assumed that the black and
  white vertices have different potential, therefore $V$ symmetry is
  broken, while $R$ and $F$ symmetries are still present.

  In part (b) the structure of the graph inside the dashed fundamental
  domain is magnified.  Gray vertices outside of the fundamental
  domain are obtained by shifts from the corresponding vertices
  inside.  For example, $f_{5'} = \omega_2 f_5$, therefore the
  operator $H(\vec{k})$ at site $2$ acts as
  \begin{equation*}
    (H(\vec{k})f)_2 = (f_2 - f_3) + (f_2 - f_1) + r(f_2 - \omega_2 f_5) + q_2 f_2,
  \end{equation*}
  where we took the longer sides in the structure of
  Fig.~\ref{fig:graph_example}(a) to have weight $1$ and the shorter
  sides weight $r$ (usually, the weight is taken to be inversely
  proportional to length).  The entire operator $H(\vec{k})$ is
  \begin{equation*}
    H(\vec{k}) =
    \begin{pmatrix}
      q_1 & -1 & 0 & r\cc{\omega_1}\omega_2 & 0 & -1 \\
      -1 & q_2 & -1 & 0 & -r\omega_2 & 0 \\
      0 & -1 & q_1 & -1 & 0 & -r\omega_1 \\
      r\omega_1\cc{\omega_2} & 0 & -1 & q_2 & -1 & 0 \\
      0 & -r\cc{\omega_2} & 0 & -1 & q_1 & -1 \\
      -1 & 0 & r\cc{\omega_1} & 0 & -1 & q_2
    \end{pmatrix},
  \end{equation*}
  with $\omega_j$
  defined in \eqref{eq:notatiok_k_omega};
  above,
  for simplicity, the potential $q$ was made to absorb the
  weighted degree of the corresponding vertex.
  
  With $q_1 = \sqrt{3}$, $q_2 = 0$ and $r=\sqrt{7}$, the eigenvalues of
  $H(\vec{k}^*)$, calculated numerically, are
  \begin{equation}\label{egn}
   -2.5097,\quad {-2.5097},\quad {-1.6753},\quad 3.4074,\quad 4.2418,\quad
    4.2418.
  \end{equation}
  
  To find the operator $Q_1$ acting on the two darker vertices in
  Fig.~\ref{fig:graph_example}(c), we use the definition of the space
  $\spaceH_1(\vec{k}^*)$, equation~(\ref{eq:12split}): for the gray
  vertices we have
  \begin{equation*}
    f_{1'} = \tau f_1, \quad f_{1''} = \cc{\tau} f_1, \quad
    f_{2'} = \tau f_2
  \end{equation*}
  by rotation and then
  \begin{equation*}
    f_{1'''} = \cc{\tau} f_{1''} = \tau f_1, \quad
    f_{2'''} = \tau f_{2'} = \cc{\tau} f_2
  \end{equation*}
  by translation (see Fig.~\ref{fig:floquet_reduction}(c) with
  $\omega_1 = \tau$ and $\omega_2 = \cc{\tau}$).  We thus
  get
  \begin{equation*}
    Q_1 =
    \begin{pmatrix}
      q_1 & -1 - \cc{\tau} - r \cc{\tau} \\
      -1 - \tau - r\tau & q_2
    \end{pmatrix}.
  \end{equation*}
  With the above choice of constants, the eigenvalues of $Q_1$ are
  \begin{equation*}
    -2.5097
\quad\mathrm{and}\quad 4.2418
  \end{equation*}
  which matches the double eigenvalues of
  $H(\vec{k}^*)$ in \eqref{egn}.  The matrices $Q_0$ and $Q_2$ can be similarly calculated
  as
  \begin{equation*}
    Q_0 =
    \begin{pmatrix}
      q_1 & -2 - \tau r \\
      -2 - \cc{\tau} r & q_2
    \end{pmatrix}
    \qquad \mbox{and} \qquad
    Q_2 =
    \begin{pmatrix}
      q_1 & -1 - \tau - r \\
      -1 - \cc{\tau} - r & q_2
    \end{pmatrix}.
  \end{equation*}
\end{example}

\begin{example}
  \label{exm:two_atom}
  We will now explain the application of our theory to the most basic
  example: the tight-binding approximation of graphene structure, with
  vertices of a discrete graph representing carbon atoms, see
  Fig.~\ref{fig:graphene_tba}(a).

  \begin{figure}[t]
    \centering
    \includegraphics[scale=0.7]{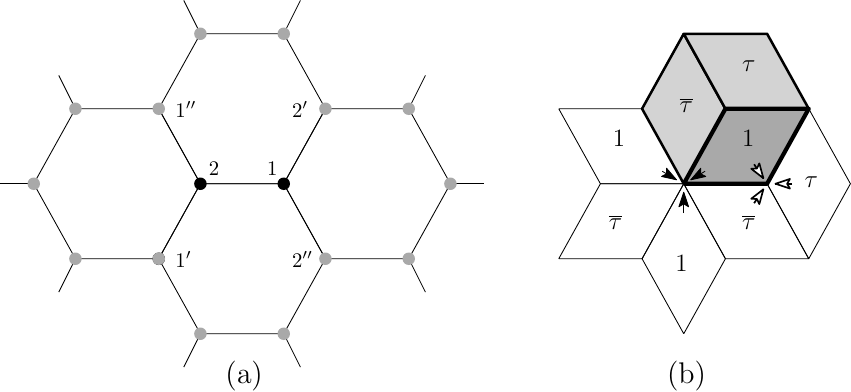}
    \caption{Graphene structure with two vertices per fundamental domain.}
    \label{fig:graphene_tba}
  \end{figure}

  The operator $H(\vec{k})$ acts on a 2-dimensional space over
  vertices $1$ and $2$ (all other vertices of the graph are obtained
  by shifts).  It acts as
  \begin{align*}
    (H(\vec{k}) f)_1 &= -f_2 - f_{2'} - f_{2''} + q f_1\\
    &= -f_2 - \omega_1 f_2 - \omega_2 f_2 + q f_1,
  \end{align*}
  and similarly for $(H(\vec{k})f)_2$.
  Note that the atoms are identical,
  hence $q_1=q_2=q$.  When $\omega_1 = \cc{\omega_2} = \tau$, the
  matrix $H$ is $q$ times identity.  

  The eigenproblem of the rhombic subdomain can be gleaned from
  Fig.~\ref{fig:graphene_tba}(b).  In particular, $f_1$ is forced to
  be zero: which can be seen from the equality
  $f_1 = \tau f_1 = \cc{\tau} f_1$ highlighted by the empty arrows in
  Fig.~\ref{fig:graphene_tba}(b), or from the boundary conditions for
  the bottom right corner of Fig.~\ref{fig:operatorsQ}(b).  On the
  other hand, the value $f_2$ is unrestricted and $(Q_1f)_2 = q f_2$.
  The complementary eigenfunction (eigenfunction of the operator
  $Q_2$) is localized on the vertex $1$.
\end{example}

%%%%%%%%%%%%%%%%%%%%%%%%%%%%%%%%%%%%%%%%%%%%%%%%%%%%%%%%%%%%%%%
%%%%%%%%%%%%%%%%%%%%%%%%%%%%%%%%%%%%%%%%%%%%%%%%%%%%%%%%%%%%%%%
\section{Conical structure around a degeneracy}
\label{sec:conical}

%%%%%%%%%%%
\subsection{General perturbation theory}

Here we list some general facts from the perturbation theory of
operators depending on parameters, following \cite{Kato,Zai+_rms75,Gru_mn09}.
Let 
\begin{equation*}
  H(r) = H_0 + (r-r_0) H_1 + O\left(|r-r_0|^2\right)  
\end{equation*}
be an analytic family of self-adjoint operators depending on one
parameter with an isolated doubly degenerate eigenvalue $\lambda_0$ at
$r=r_0$.  The eigenvalue then splits into two analytic branches
\begin{equation*}
  \lambda^{\pm}(r) = \lambda_0 + \lambda_1^{\pm} (r-r_0) 
  + O\left(|r-r_0|^2\right).
\end{equation*}
The linear terms can be found as the eigenvalues of the $2\times2$
matrix $P H_1 P$, where $P$ is the projector onto the eigenspace of
$\lambda_0$.  The corresponding eigenvectors expand as
\begin{equation}
  \label{eq:pert_eigenvectors}
  \psi^{\pm}(r) = \psi^{\pm}_0 
  + O\left(\frac{|r-r_0|}{|\lambda_1^+ - \lambda_1^-|}\right),
\end{equation}
where $\psi_0^{\pm}$ are the eigenvectors of $P H_1 P$ (which are in
the eigenspace of $H_0$).  All eigenvectors are assumed to be normalized.

If $H = H(k_1,k_2)$ is an analytic function of two parameters and
$\vec{k}_0$ is the point of double multiplicity of the eigenvalue $\lambda_0$, the
one-parameter theory is still valid in every direction $\delta k_1 = r
\cos(\phi)$, $\delta k_2 = r \sin(\phi)$.  The parameters $\lambda_1^{\pm}$
now depend on the direction $\phi$.

We will say that a doubly degenerate eigenvalue is a conical point if
$\lambda_1^+(\phi) \neq \lambda_1^-(\phi)$ in every direction; more
precisely,

\begin{definition}
  Let $H(\vec{k})$ be an analytic family of self-adjoint operators.
  We will say that $H(\vec{k})$ has a \emph{nondegenerate conical point} at
  $\vec{k}_0$ with an eigenvalue $\lambda_0$ if
  $\lambda_0\in\sigma\sb{d}(H(\vec{k}_0))$ is an isolated eigenvalue
  of geometric multiplicity $2$, and in an
  open neighborhood of $\vec{k}_0$
  the eigenvalues are given by
  \begin{equation}
    \label{eq:nondeg_splitting}
    \lambda^{\pm}(\vec{k})
    =\lambda_0+\delta\vec{k}\cdot\vec{n}
    \pm\sqrt{Q(\delta\vec{k})} + O(\abs{\delta\vec{k}}^2),
    \qquad \delta\vec{k} = \vec{k} - \vec{k}_0,
  \end{equation}
  where $\vec{n}\in\R^2$ and
  $Q(\vec{k})$ is a positive-definite quadratic form.  The point
  $\vec{k}_0$ is a \emph{fully degenerate conical point} if the same
  is true with $Q \equiv 0$.
\end{definition}

From Lemma~\ref{lem:symm_disp} we know that the points of double
degeneracy at $\pm \vec{k}^*$ and $\vec{0}$ must either be
nondegenerate circular cones (in $\kappa$ coordinates) or fully
degenerate cones.  It turns out that the point $\vec{0}$ is always a
fully degenerate cone; we will also derive a condition for
nondegeneracy of the cone at $\pm \vec{k}^*$.

In the first order of perturbation theory (i.e.\ ignoring the
$O(|\vec{\delta k}|^2)$ term in (\ref{eq:nondeg_splitting})), the dispersion
surface is given by the solution to\footnote{This is a standard
  procedure in quantum mechanics or solid state physics (known as
  $k\cdot p$ method in the latter); for a mathematical proof, see
  \cite{Gru_mn09}.}
\begin{equation}
  \label{eq:det_dispersion}
  \det \left( \delta k_1\, h_{1} + \delta k_2\, h_{2} - (\lambda-\lambda_0) \right) = 0,
\end{equation}
where the $2\times2$ Hermitian matrices $h_1$ and $h_2$ are given by
\begin{equation}
  \label{eq:deriv_matrices}
h_{j}
= \Phi^* \frac{\partial H}{\partial k_j} \Phi
=
\begin{bmatrix}
\langle f_1,\frac{\partial H}{\partial k_j} f_1\rangle
&
\langle f_1,\frac{\partial H}{\partial k_j} f_2\rangle
\\
\langle f_2,\frac{\partial H}{\partial k_j} f_1\rangle
&
\langle f_2,\frac{\partial H}{\partial k_j} f_2\rangle
\end{bmatrix},
\qquad
j=1,\,2.
\end{equation}
Here
$\Phi=[f_1,f_2]$ is a matrix whose columns are
the orthonormal basis vectors of the degenerate eigenspace
at $(0,0)$:
\[\Phi : \C^2 \to \spaceH,
\qquad
\Phi:
\begin{bmatrix}c_1\\c_2\end{bmatrix}
\mapsto
c_1 f_1+c_2 f_2.
\]
The projector $P$ onto the eigenspace is then given by $P
= \Phi\Phi^*$.

%%%%%%%%%%%
\subsection{Perturbation in the presence of symmetry}
\label{sec:pert_symm}

Naturally, the presence of symmetry imposes constraints on the form of
the matrices $h_{1}$ and $h_{2}$.  As we will see in Lemma XX and YY
below, these constraints are often powerful enough to give an explicit
form of the dispersion relation.

\begin{lemma}
  \label{lem:intertwiner}
  Let $H(\vec{k})$ be an analytic family of self-adjoint operators and
  the unitary operator $S$ satisfy
  \begin{equation}
    \label{eq:daul_action}
    S H(\vec{k}) S^* = H(\hat{S} \vec{k}),
  \end{equation}
  where the matrix $\hat{S}$ encodes the action of $S$ on the dual
  space.  Let $\vec{k}_0$ be a fixed point of $\hat{S}$ and
  $[f_1, \ldots, f_m]$ be an orthonormal basis of an eigenspace of
  $H(\vec{k}_0)$.  Let the unitary matrix $A_S$ encode the action of
  $S$ in this basis, namely
  \begin{equation}
    \label{eq:AS_def}
    S \Phi = \Phi A_S, \qquad \mbox{where }\Phi:\C^m\to \spaceH,
    \ \Phi \vec{c} = c_1 f_1 + \ldots c_m f_m.
  \end{equation}
  Then $ h_{\vec{\delta k}}
  := \delta k_1\, h_1 + \delta k_2\, h_2$ satisfies
  \begin{equation}
    \label{eq:intertwiner_unit}
    A_S h_{\vec{\delta k}} A_S^* = h_{\hat{S}\vec{\delta k}}.
  \end{equation}

  If $S$ is an antiunitary operator satisfying (\ref{eq:daul_action})
  and $S \Phi = \Phi A_S C$, then
  \begin{equation}
    \label{eq:intertwiner_antiunit}
    A_S h_{\vec{\delta k}} A_S^* = \cc{ h_{\hat{S}\vec{\delta k}} }.
  \end{equation}
\end{lemma}

\begin{proof}
  From equation~(\ref{eq:daul_action}) we have
  \begin{equation*}
    S \left(H(\vec{k}_0 + \vec{\delta k}) - H(\vec{k}_0) \right) S^*
    = H\left(\hat{S}(\vec{k}_0+\vec{\delta k})\right)  - H\left(\hat{S}\vec{k}_0\right)
    = H\left(\vec{k}_0 + \hat{S}\vec{\delta k}\right)  - H\left(\vec{k}_0\right),
  \end{equation*}
  where we used the fact that $\vec{k}_0$ is a fixed point of
  \eqref{eq:rot_mod_map}.  Passing to the limit, we get
  \begin{equation}
    \label{eq:deriv_conj_S}
    S \left( \delta k_1 \left.\frac{\p H}{\p k_1}\right|_{\vec{k}_0} 
      + \delta k_2 \left.\frac{\p H}{\p k_2}\right|_{\vec{k}_0} \right) S^* 
    =: S \left( D_{\vec{\delta k}}H \right) S^*
    = D_{\hat{S}\vec{\delta k}} H,
  \end{equation}
  where $D_{\vec{\delta k}} H$
  % \vec{k}\cdot\nabla\sb{\vec{k}}H
  % =
  is the directional
  derivative of $H$ in the direction $\vec{\delta k}$ at the point $\vec{k}_0$.
 
  Note that $h_{\vec{\delta k}} = \Phi^* \left({D_{\vec{\delta k}} H}\right) \Phi$.
  Conjugating equation \eqref{eq:deriv_conj_S} by the matrix $\Phi$,
  we get~(\ref{eq:intertwiner_unit}) (note that $A_{S^*} = A_S^*$).
  The antiunitary case is analogous.
\end{proof}

%%%%%%%%%%%%
\subsection{Application to graphene operators}
\label{sec:appl_pert}

\begin{lemma}
  \label{lemma:structure_of_Hderiv}
  Let the self-adjoint operator $H$ be $\Gamma$-periodic and invariant
  under rotation $R$.  If $H(\vec{k}^*)$ has an eigenvalue $\lambda_0$
  of multiplicity two with eigenvectors satisfying
  \begin{equation}
    \label{eq:rot_eigenfunctions}
    R f_1 = \tau f_1, \qquad R f_2 = \cc{\tau} f_2, \qquad \tau = e^{2\pi i / 3},
  \end{equation}
  the dispersion relation has the form 
  $\lambda - \lambda_0 = \pm |\alpha| |\vec{\kappa} - \vec{\kappa}_0|
  + O(\abs{\vec{\kappa}-\vec{\kappa}\sb 0}^2)$,
  with 
  \begin{equation}
    \label{eq:alpha_formula}
    \alpha =
    \left\langle f_1, \frac{\partial H}{\partial \kappa_1} f_2 \right\rangle.
  \end{equation}
\end{lemma}

\begin{remark}
  This calculation was performed for $\R^2$ Laplacian with any
  $R$-symmetric potential in \cite[Prop 4.2]{FefWei_jams12}, using
  explicit calculation of the derivatives
  ${\partial H}/{\partial \kappa_j}$.  We show that it is a direct
  corollary of Lemma~\ref{lem:intertwiner}.
\end{remark}

\begin{proof}
  We use Lemma~\ref{lem:intertwiner} with the symmetry $S=R$.
  From (\ref{eq:rot_eigenfunctions}) we obtain
  \begin{equation}
    \label{eq:AR}
    A_R = 
    \begin{pmatrix}
      \tau & 0 \\ 0 & \cc{\tau}
    \end{pmatrix}.
  \end{equation}
  Using the explicit form of the matrix $\hat{R}$ from
  \eqref{eq:hatR_def}, equation~(\ref{eq:intertwiner_unit}) can
  be written in components as
  \begin{equation}
    \label{eq:deriv_h1_cond}
    \begin{pmatrix}
      \tau & 0 \\ 0 & \cc{\tau}
    \end{pmatrix}
    h_1
    \begin{pmatrix}
      \cc{\tau} & 0 \\ 0 & \tau
    \end{pmatrix}
    = -h_1 -h_2, 
    \qquad 
    \begin{pmatrix}
      \tau & 0 \\ 0 & \cc{\tau}
    \end{pmatrix}
    h_2
    \begin{pmatrix}
      \cc{\tau} & 0 \\ 0 & \tau
    \end{pmatrix}
    = h_1.
  \end{equation}
  It is now straightforward to check that any $2\times2$ Hermitian
  matrices satisfying \eqref{eq:deriv_h1_cond} must be of the form
  \begin{equation}
    \label{eq:deriv_rot}
    h_{1} =
    \begin{pmatrix}
      0 & \beta \\
      \cc\beta & 0
    \end{pmatrix},
    \qquad
    h_{2} =
    \begin{pmatrix}
      0 & \tau\beta \\
      \cc\tau \cc\beta & 0
    \end{pmatrix}, 
    \qquad
    \mbox{where }
    \beta = \left\langle f_1, \frac{\partial H}{\partial k_1} f_2 \right\rangle.
  \end{equation}
  We now calculate the shape of the dispersion relation in the first
  order of perturbation theory using \eqref{eq:det_dispersion}.  It is
  \begin{equation}
    \label{eq:disp_conical}
    (\lambda-\lambda_0)^2 - |\beta|^2 |\delta k_1 + \tau\, \delta k_2|^2 
    = (\lambda-\lambda_0)^2 - \frac{3}{4}|\beta|^2 |\delta\kappa|^2 = 0,
  \end{equation}
  where we changed to the coordinates $\vec{\kappa} = k_1 \vec{b_1} +
  k_2\vec{b_2}$ in which the dispersion relation is the circular cone
  with no tilt.  To relate the answer to \eqref{eq:alpha_formula}, we
  observe that
  \begin{equation*}
    \frac{\partial H}{\partial \kappa_1} = \frac{\sqrt{3}}2 \left(
      \frac{\partial H}{\partial k_1} + \frac{\partial H}{\partial
        k_2}
      \right), 
  \end{equation*}
  and therefore, from \eqref{eq:deriv_rot}, $\alpha =
  \frac{\sqrt{3}}2(1+\tau) \beta$.  Since $|\alpha|^2 = \frac34
  |\beta|^2$, we get the promised answer.
\end{proof}

The cone becomes degenerate if $\alpha=0$ (this condition is
equivalent to condition (4.1) of \cite{FefWei_jams12}).  In
\cite{Gru_mn09}, $\alpha$ was shown to be non-zero for small
$\epsilon$ in $H_\varepsilon=-\Delta + \varepsilon q(\vec{x})$; by
analyticity, the cone can be degenerate only for isolated values of
the parameter $\epsilon$.  We explore this in more detail in the next
section.

%%%%%%%%%%
\subsection{Perturbation of the pure Laplacian}
\label{sec:pert_pure}

In this section we describe in more detail the case of Laplacian on
$\R^2$ with the bounded potential $q(\vec{x})$ considered as a
perturbation, $H_\varepsilon=-\Delta + \varepsilon q(\vec{x})$.
%\[
%H_\varepsilon( \vec{k} ) = -(\nabla + \vec{k})^2 + epsilon q(\vec{x}),
%\]
%%(2.15) in Fefferman-Weinstein
%%Grushin: beginning of sect 4
Similar calculation appeared in \cite{Gru_mn09} and
\cite{FefWei_jams12} (see also \cite{FefLeeWei_prep14}), therefore we
concentrate on connections with the results presented above.

\begin{proof}[Proof of Theorem~\ref{thm:Laplacian}]
  When $\varepsilon=0$, the lowest eigenvalue of $H_0(\vec{k}^*)$ is triply
  degenerate.  Indeed, the function
  \begin{equation}
    \label{eq:eigenvector_of_Lapl}
    \phi(\vec{x}) := \exp\left(i \vec{\kappa}^* \cdot \vec{x}\right)
    = \exp\left(\frac{4\pi i}{3} x_2\right)
  \end{equation}
  is an eigenfunction of the Laplacian and satisfies
  \begin{equation*}
    \phi(\vec{x}+\vec{a}_1) = \tau \phi(\vec{x}), \qquad
    \phi(\vec{x}+\vec{a}_2) = \cc{\tau} \phi(\vec{x}),
  \end{equation*}
  therefore it is an eigenfunction of $H_0(\vec{k}^*)$.  Since $R$, the
  operator of rotation by $2\pi/3$, commutes with $H_0(\vec{k}^*)$, the
  functions 
  \begin{equation}
    \label{eq:Rphi}
    R \phi = \exp\left(\frac{4\pi i}{3}
      \left(-\frac{\sqrt3}{2}x_1-\frac12 x_2\right)\right),
    \qquad
    R^2 \phi = \exp\left(\frac{4\pi i}{3}
      \left(\frac{\sqrt3}{2}x_1-\frac12 x_2\right)\right),  
  \end{equation}
  are also eigenfunctions.  It can be
  verified directly that they are orthogonal.  Their combinations
  \begin{equation}
    \label{eq:Q_eigenfunction}
    \psi_j(\vec{x}) := \frac{1}{3}\left( \phi(\vec{x}) + \tau^j R \phi(\vec{x}) +
      \cc{\tau}^j R^2 \phi(\vec{x})\right) =: P_j\phi, \quad j=0,1,2,
  \end{equation}
  are simple eigenfunctions of the operator $H_0(\vec{k}^*)$
  restricted to $\spaceH_j(\vec{k}^*)$ for $j=0,1,2$ correspondingly.

  We now need to show that the eigenvalues of $H$ in
  $\spaceH_0(\vec{k}^*)$ and $H$ in $\spaceH_1(\vec{k}^*)$ (or
  $\spaceH_2(\vec{k}^*)$) will separate for non-zero $\epsilon$
as long as \eqref{eq:cond_blah} is satisfied.  In
  the first order perturbation theory, the condition for separation is
  \begin{equation}
    \label{eq:sep_cond}
    \frac{\langle\psi_0, q(\vec{x})\psi_0\rangle}{\langle\psi_0,
      \psi_0\rangle} 
    \neq \frac{\langle\psi_1, q(\vec{x}) \psi_1\rangle}{\langle\psi_1,
      \psi_1\rangle}
    = \frac{\langle\psi_2, q(\vec{x}) \psi_2\rangle}{\langle\psi_2,
      \psi_2\rangle}
  \end{equation}
  where the scalar products are taken in $L^2(\Omega_R)$.  Since
  $\|\psi_0\|=\|\psi_1\| = \|\psi_2\|$ and $\tau+\cc{\tau}=-1$,
  condition (\ref{eq:sep_cond}) is equivalent to
  \begin{equation}
    \label{eq:sep_cond2}
    \left\langle P_0\phi, q(\vec{x})P_0\phi\right\rangle_{L^2(\Omega_H)} 
    + \tau \left\langle P_1\phi, q(\vec{x})P_1\phi\right\rangle_{L^2(\Omega_H)}
    + \cc\tau \left\langle P_2\phi, q(\vec{x})P_2\phi\right\rangle_{L^2(\Omega_H)}
    \neq 0.
  \end{equation}
  Using that $P_j$ are projectors which commute with multiplication by
  the $R$-invariant function $q(x)$, we reduce the left-hand side to 
  \begin{equation}
    \label{eq:sep_cond3}
    \left\langle (P_0+\cc\tau P_1 + \tau P_2)\phi, q(\vec{x})\phi\right\rangle 
    = \langle R\phi, q\phi \rangle = \langle R^2\phi, q R\phi \rangle
    = \int_{\Omega_H} e^{\frac{4\pi i}{\sqrt3} x_1} q(\vec{x})\, d\vec{x},
  \end{equation}
in agreement with \eqref{eq:cond_blah}.
 
  Two more facts are now needed to establish existence of non-degenerate
  conical points for almost all values of $\varepsilon>0$.
  \begin{enumerate}
  \item
    The parameter
    $\alpha$ describing the opening angle of the cone, see
    equations~(\ref{eq:alpha_formula}) and (\ref{eq:disp_conical}), is
    analytic as a function of $\varepsilon$.
  \item
    $\alpha$ is nonzero when $\varepsilon=0$.
  \end{enumerate}
  
  %%L1loc
  Analyticity of $\alpha=\alpha(\varepsilon)$ follows from the
  analyticity of the eigenfunction corresponding to a simple
  eigenvalue of the self-adjoint operator $H(\varepsilon)$ on the
  fixed space $\spaceH_1(\vec{k}^*)$ as a function of one parameter;
  this is a consequence of the results of Rellich and Kato, see
  \cite[Sec.~VII.3]{Kato} and \cite{Rel_ma39}.  The corresponding
  eigenfunction $f_1$ is also analytic and so is $f_2$.  The
  derivative $\partial H / \partial k_1 =\partial H_0/\partial k_1$
  does not depend on $\varepsilon$,
  %% BECAUSE:
  %% AC H_epsilon( k ) = -(\nabla + \vec{k})^2 + epsilon q(\vec{x})
  %% 
  therefore $\alpha$ defined by (\ref{eq:alpha_formula}) is analytic.

  Finally, we calculate the value of $\alpha(0)\neq0$ explicitly.  By
  the standard gauge transformation technique,
  $D_{\vec{\kappa}} H = -2i\vec{\kappa}\cdot \nabla$.  Therefore, using
  (\ref{eq:Q_eigenfunction}) and orthogonality of $\phi$, $R\phi$ and
  $R^2\phi$, we get
  \begin{align}
    \label{eq:alpha_calc}
    \alpha &= \frac{1}{\|\psi_1\| \|\psi_2\|}
    \left\langle \psi_1, \frac{\partial H}{\partial \kappa_1}
      \psi_2 \right\rangle_{L^2(\Omega_H)}
    = \frac{-2i}{\|\psi_1\| \|\psi_2\|} \left\langle \psi_1, \frac{\partial}{\partial x_1}
      \psi_2 \right\rangle
\nonumber
\\
    &= \frac{-2i}{3\|\phi\|^2} \left\langle \phi + \tau R\phi + \tau^2 R^2\phi, 
      \frac{\partial}{\partial x_1}
      \left(\phi+\tau^2 R\phi+\tau R^2\phi\right) \right\rangle
\nonumber
\\
    &= \frac{-2i}{3\|\phi\|^2} \left\langle \phi + \tau R\phi + \tau^2 R^2\phi, 
      -\frac{2\pi i}{\sqrt{3}} \tau^2 R\phi 
      + \frac{2\pi i}{\sqrt{3}}\tau R^2\phi \right\rangle
    = \frac{4\pi}{3\sqrt{3}}(-\tau + \tau^2) = -\frac{4\pi i}{3}.
  \end{align}
\end{proof}

\begin{remark}
  The assumption $q\in L^\infty(\R^2)$ could be relaxed.
%% We note that
%% the discreteness of spectrum
%% for periodic potentials $V\in L\sb{\mathrm{loc}}^1(\R)$
%% in the one-dimensional case
%% follows from the proof of \cite[Lemma 2.9]{MR2350362},
%% which shows that the resolvent is a Hilbert--Schmidt operator.
The discreteness of spectrum and analyticity of eigenvalues
of $H(\vec{k})$ (as functions of quasi-momenta $\vec{k}$)
for periodic potentials $q\in L\sb{\mathrm{loc}}^{1+\epsilon}(\R^2)$,
$\epsilon>0$,
follows from the argument in \cite[Theorem 3.1]{AvrSim_ap78}
(where the corresponding result is obtained for the three-dimensional case
when $q\in L\sb{\mathrm{loc}}^{3/2}(\R^3)$).
Under this assumption,
the potential $q$ is a relatively bounded perturbation
with relative bound zero
and $H(\vec{k})$ is analytic family of type B
in the sense of Kato~\cite{Kato}.
\end{remark}

\begin{remark}
  Consider a potential $q(\vec{x})$ which is $R$-invariant, but may
  not have
  $V$ or $F$ symmetry.  It can be shown that the first order
  perturbation condition for the eigenvalues of $Q_1$ and $Q_2$ to
  \emph{not separate} is precisely that the right hand side of
  equation \eqref{eq:sep_cond3} is real.  The latter is of course
  satisfied if $q(\vec{x})$ does have $V$ or $F$ symmetry.
\end{remark}

\begin{example}
  To continue with Example~\ref{exm:graph_structure}, it is
  interesting to investigate\footnote{This question was asked by
    P.~Kuchment.} what happens when the parameter $r$ is equal to 1.
  At the special points $\pm\vec{k}^*$ there are now triple
  degeneracies as the spectrum of $Q_0$ coincides at this point with
  the spectra of $Q_1$ and $Q_2$.  As a consequence there is no
  conical point there.  Instead, the lower 3 sheets of the dispersion
  relation develop singularities along curves and touch each
  other to form an intricate picture, Figure~\ref{fig:dispersion3deg}.
  The picture can be resolved as three analytic surfaces crossing each
  other.  Similar shape is assumed by the upper 3 sheets.

  \begin{figure}[t]
    \includegraphics[scale=0.35]{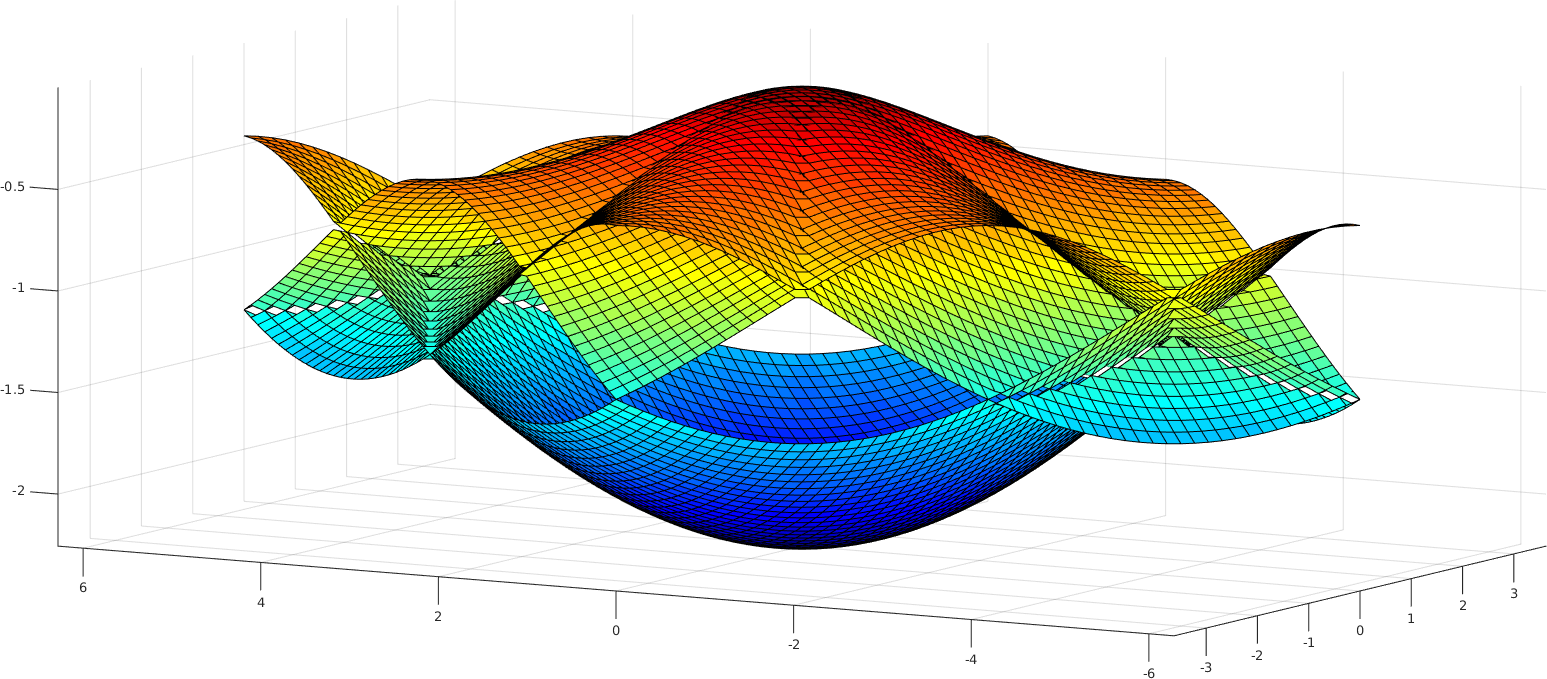}
    \caption{The lower three sheets of the dispersion relation for the
      graph in Example~\ref{exm:graph_structure} with the parameter
      $r=1$.}
    \label{fig:dispersion3deg}
  \end{figure}  

  The reason for such a complicated picture is that the system now has
  more symmetry and the three sheets can be obtained by (1)
  % \ac{to be reformulated}
  considering the smaller fundamental domain, (2) cutting up its
  dispersion relation and folding it back to Brillouin zone chosen in
  Figure~\ref{fig:dispersion3deg}.  This is analogous to the situation
  with $H_0 = -\Delta$ above which has more symmetry than the
  hexagonal lattice.  It also illustrates the observation of
  \cite{FefWei_jams12} that the cones may degenerate at isolated
  values of a parameter ($r$, in the present example).
\end{example}

%%%%%%%%%%%%%%%%%%%%%%%%%%%%%%%%%%%%%%%%%%%%%%%%%%%%%%%%%%%%%%%
%%%%%%%%%%%%%%%%%%%%%%%%%%%%%%%%%%%%%%%%%%%%%%%%%%%%%%%%%%%%%%%
\section{Degeneracy at $\vec{k} = \vec{0}$}
\label{sec:point0}

The third fixed point of the rotation $\hat R$ in the momentum space
(see Lemma~\ref{lemma:sym_disp}) also leads to degeneracies in the
spectrum.  They are present even if both inversion and reflection
symmetries are broken: rotation and complex conjugation are sufficient
to retain degeneracies.  However, the local structure of the
dispersion relation is a \emph{degenerate} cone, see
Fig.~\ref{fig:dispersion} for an example.

\begin{lemma}
  \label{lemma:point0}
  Let the self-adjoint operator $H$ be $\Gamma$-periodic and invariant
  under rotation $R$.  The space $\spaceH(\vec{0})$ splits into the
  orthogonal sum
  \begin{equation}
    \label{eq:split_spaceH_star3}
    \spaceH(\vec{0}) = \spaceH_0(\vec{0}) \oplus
    \spaceH_\perp(\vec{0}),
  \end{equation}
  where
  $\spaceH_0(\vec{0}) = \{\psi \in \spaceH(\vec{0}) : R\psi =
  \psi\}$. This splitting is $H$-invariant.  

  If $H$ is also invariant with respect to complex conjugation, then
  all eigenvalues of the operator $H$ restricted to
  $\spaceH_\perp(\vec{k}^*)$ have even multiplicity.  Moreover, each
  eigenspace has an orthonormal basis $\{f^1_n, f^2_n\}$, such that
  \begin{equation}
    \label{eq:spec_basis_RC}
    R f^1_n = \tau f^1_n, \qquad 
    R f^2_n = \cc{\tau}f^2_n, \qquad
    \mbox{and} \qquad
    f^2_n = \cc{f^1_n}.     
  \end{equation}

  If $\lambda = \lambda_0$ is an eigenvalue of multiplicity two, then
  the dispersion relation is locally flat at $\vec{k} = \vec{0}$:
  \begin{equation}
    \label{eq:flat_disp}
    \lambda - \lambda_0 = O(|\vec{k}|^2).
  \end{equation}
\end{lemma}

\begin{remark}
  The eigenvalue $\lambda_1(\vec{0})$ is always non-degenerate,
  therefore first and second bands cannot touch at $\vec{k} = \vec{0}$. 
\end{remark}

\begin{proof}
  The proof of the first part is identical to the proof of
  Lemma~\ref{lemma:mult} in the case of symmetries $R$ and $\Vbar$.

  To prove the estimate~\eqref{eq:flat_disp}, we use the special basis
  satisfying \eqref{eq:spec_basis_RC}.  The proof of
  Lemma~\ref{lemma:structure_of_Hderiv} still applies so the matrices
  $h_1$ and $h_2$ have the form given by \eqref{eq:deriv_rot}.
  Applying Lemma~\ref{lem:intertwiner} to the complex conjugation $C$
  as an antiunitary symmetry of
  $H(\vec{k})$ at $\vec{k}=\vec{0}$, and using
  \eqref{eq:spec_basis_RC}, we get
  \begin{equation}
    \label{eq:deriv_matrix_conj_conj}
    \begin{pmatrix}
      0 & 1\\ 1 & 0
    \end{pmatrix} h_{\vec{k}}
    \begin{pmatrix}
      0 & 1\\ 1 & 0
    \end{pmatrix}
    = \cc{ h_{-\vec{k}} }.
  \end{equation}
  This is consistent with \eqref{eq:deriv_rot} only if $h_{\vec{k}}\equiv0$.
  Then, according to (\ref{eq:det_dispersion}), $\lambda-\lambda_0 =
  O(|\vec{k}|^2)$
  yielding the conclusion.
\end{proof}

\begin{remark}
  More generally, one can get the following result.
  Suppose the operator $H(\vec{k}) = H(k_1,k_2)$ has the following symmetry at
  the point $\vec{k}_0$:
  \begin{equation*}
    H\left(\vec{k}_0 - \vec{k}\right)
    =
    \overline{H\left(\vec{k}_0 + \vec{k}\right)}
    := C H\left(\vec{k}_0 + \vec{k}\right) C^{-1}.
  \end{equation*}
  If $\lambda_0$ is an eigenvalue of $H(\vec{k}_0)$ of multiplicity 2,
  it cannot be a nondegenerate conical point.  In the leading order,
  it must have the form of two intersecting planes of which
  \eqref{eq:flat_disp} is a degenerate example.
\end{remark}

\begin{figure}[t]
  \centering
  \includegraphics[scale=0.6]{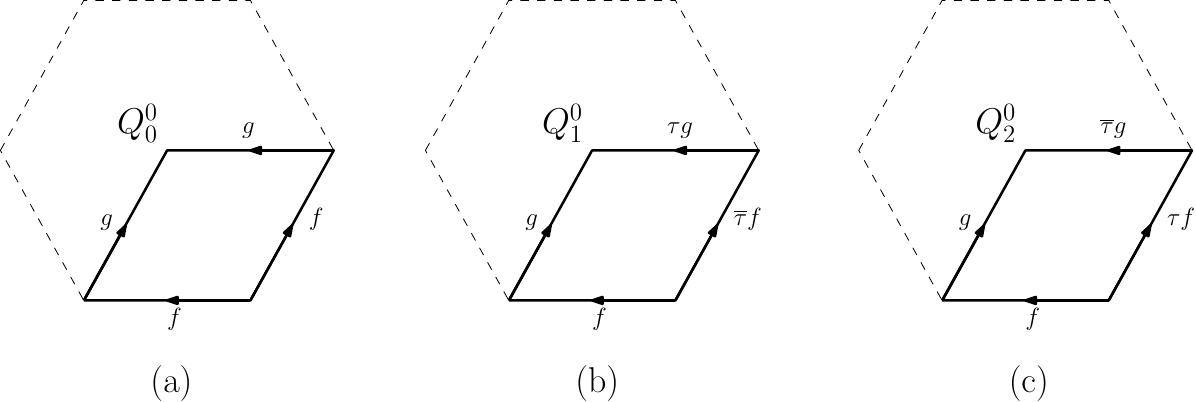}
  \caption{The operators $Q^0_j$ that together give the spectrum of
    $H(\vec{0})$.} 
\label{fig:operatorsQ0}
\end{figure}

\begin{remark}
  We can also obtain a more detailed splitting of $H(\vec0)$ similar
  to Section~\ref{sec:full_split}, into a sum of three operators,
  $Q^0_j$, $j=0,1,2$, whose boundary conditions are described in Figure~\ref{fig:operatorsQ0}.
\end{remark}

\begin{example}
  \label{exm:graph_ex_0}
  Revisiting Example~\ref{exm:graph_structure} and calculating the
  eigenvalues of $H(\vec{0})$ numerically, we get
  \begin{equation*}
    -3.8598, \quad {-0.9937}, \quad {-0.9937}, \quad 2.7257, \quad 2.7257,
    \quad 5.5918.
  \end{equation*}
  The corresponding operator $Q^0_1$ in this case can be shown to be
 \begin{equation*}
    Q^0_1 =
    \begin{pmatrix}
      q_1 & -1 - \cc{\tau} - r \tau \\
      -1 - \tau - r\cc\tau & q_2
    \end{pmatrix}.
  \end{equation*}
  with eigenvalues $-0.9937$ and $2.7257$.

  Interestingly, in the case of Example~\ref{exm:two_atom}, the graph
  structure is not rich enough to support the operators $Q^0_1$ or $Q^0_2$: it can
  be shown that in this case $\spaceH_\perp(\vec{0}) = \emptyset$.
\end{example}

In Appendix~\ref{sec:pert_k0} we give a brief account of the case of
pure Laplacian on $\R^2$ at the quasi-momentum point $\vec{k} = 0$.  It
is largely parallel to Section~\ref{sec:pert_pure}, but requires 
delving deeper into 
representation theory of Section~\ref{sec:representations}.

%%%%%%%%%%%%%%%%%%%%%%%%%%%%%%%%%%%%%%%%%%%%%%%%%%%%%%%%%%%%%%%%%%%%%%
%%%%%%%%%%%%%%%%%%%%%%%%%%%%%%%%%%%%%%%%%%%%%%%%%%%%%%%%%%%%%%%%%%%%%%

\section{Persistence of conical points}
\label{sec:persistence}

We are now going to study the fate of the conical point when
the rotational symmetry is broken by a small perturbation.  We will
consider two cases:
when the perturbation retains the
conjugate inversion symmetry $\Vbar$
and when it retains the reflection symmetry $F$
(all other symmetries may or may not be
broken).  In both cases the conical point survives.  Moreover, in the
second case we are able to restrict the location of the surviving
point to a line in $\vec{k}$ space.  Of course, if the perturbation
retains both symmetries, $\Vbar$ and $F$, the stronger second result still
applies.

%%%%%%%%%
\subsection{Keeping $\Vbar$ symmetry: Berry phase}
\label{sec:keepingV}

Let us consider the weakly broken $R$ symmetry: we add to $H$ a
perturbation which is $\Vbar$-invariant but not $R$-invariant.  The $F$
symmetry may or may not be preserved.

The tool for proving Theorem~\ref{thm:persistence} in this case is the
``Berry phase'' \cite{Ber_prsla84,Sim_prl83} (also known as
``Pancharatnam--Berry phase'' or ``geometric phase''), of which we
first give an informal description.  Consider choosing a closed
contour in the parameter space and tracking certain eigenvalue along
this contour.  The eigenvalue changes as we move along the contour,
but we assume it remains simple.  Now we choose the corresponding
normalized eigenvector at every point of the contour.  The eigenvector
is defined up to a phase, and we choose it ``in the most continuous
fashion''.  Once we completed the loop, the final eigenvector must
equal the initial eigenvector up to a phase factor $e^{i \phi}$.  The
phase $\phi$ we call the Berry phase.  The fact that it might be
different from zero (modulo $2\pi$) in the simplest form of real
operator $H$ and a contour encircling a conical point has been known
for a while, see \cite{HerLon_dfs63} and \cite[Appendix
10.B]{Arnold_mat_methods}.

We now argue that the Berry phase of the operator
$H_\epsilon(\vec{k})$
%%AC is restricted to only
can only take values $0$ or $\pi$ (modulo $2\pi$).
Because of the symmetry of the
perturbation $W$, the perturbed operator $H_\epsilon(\vec{k})$ will
retain the symmetry $\Vbar$ for all $\vec{k}$.  The operator
$\Vbar$ is an antiunitary involution, i.e.
\begin{equation}
  \label{eq:def_antiunitary_inv}
  \Vbar(\alpha v) = \cc{\alpha} (\Vbar v),
  \qquad\Vbar^2 = 1, 
  \qquad
\langle\Vbar v, \Vbar u\rangle
= \langle u,v\rangle.
\end{equation}
If $\psi$ is a simple eigenfunction of $H(\vec{k})$, then, after
multiplication by a suitable phase,
\begin{equation}
  \label{eq:rep_Vbar}
  \Vbar\psi := \overline{\psi(-\vec{x})} = \psi.
\end{equation}
Indeed, because $\Vbar$ commutes with the operator
$H(\vec{k})$, $\overline{\psi(-\vec{x})}$ is an eigenvector with the
same eigenvalue and thus equal to $e^{i\theta} \psi$ for some
$\theta$.  Multiplying $\psi$ by $e^{i\theta/2}$ makes it satisfy
equation~\eqref{eq:rep_Vbar}.  
% In an easy generalization, a degenerate
% eigenspace will have a basis (over $\C$) of the vectors satisfying
% \eqref{eq:rep_Vbar}.

Condition~\eqref{eq:rep_Vbar} gives us a canonical way to choose the
overall phase of the eigenvector, up to a sign.\footnote{This choice of the
eigenvector along a curve in the parameter space
defines a parallel section of the line bundle of the eigenspaces.}
%   Indeed, the
% section $v(t)$ is parallel if
% \begin{equation*}
%   (v(t), v(t+\delta t)) = 1 + o(\delta t).
% \end{equation*}
% In our case we have, on one hand,
% \begin{equation*}
%   (\psi(t_1), \psi(t_2)) = (V\psi(t_1), V\psi(t_2)) 
%   = \overline{(\Vbar\psi(t_1), \Vbar\psi(t_2))}
%   = \overline{(\psi(t_1), \psi(t_2))} \in \R,
% \end{equation*}
% since $\psi(t)$ satisfy~\eqref{eq:rep_Vbar}.  On the other hand, by
% a polarization identity,
% \begin{align*}
%   (\psi(t), \psi(t+\delta t)) &= \Re (\psi(t), \psi(t+\delta t))\\
%   &= \frac12 \left(\|\psi(t)\| + \|\psi(t+\delta t)\| - \|\psi(t+\delta
%     t) - \psi(t)\| \right) = 1 + O(\delta t^2).
% \end{align*}
Now consider a closed path in the parameter $\vec{k}$ space.  The
phase acquired by a parallel section of the eigenspaces
(the formal
definition of the Berry phase) is restricted by condition
\eqref{eq:rep_Vbar}: the factor must be either $+1$
or $-1$, so the phase is either $0$ or $\pi$ modulo $2\pi$.

On the other hand, the phase must change continuously upon a
continuous deformation of the contour.  Therefore, if the contour is
homotopically equivalent to a point (i.e.\ encloses no parameter values
where the eigenvalue becomes multiple), the phase must be equal to
zero.  But if the contour encloses a conical point, the phase is equal
to $\pi$ modulo $2\pi$.

\begin{lemma}
  \label{lem:minus_one}
  Let the self-adjoint operator $H(\vec{k})$, which analytically
  depends on the two parameters $\vec{k} = (k_1,k_2)$, have a
  nondegenerate conical point at $(0,0)$.  Let $H(\vec{k})$ commute
  with an antiunitary involution $\Vbar$.  Then the Berry phase
  acquired on a contour enclosing the singularity $(0,0)$ is $\pi$.
\end{lemma}

\begin{remark}
  This result for a real-valued operator $H$ can be traced back at
  least to Herzberg and Longuet-Higgins \cite{HerLon_dfs63}.  Their
  proof is based on reducing the question using perturbation theory to
  a question about $2\times2$ matrices and computing the eigenvectors
  explicitly.  A more general formula is derived in
  \cite[Sec.~3]{Ber_prsla84}, from which Lemma~\ref{lem:minus_one}
  follows.  In Appendix~\ref{sec:berry} we include an alternative
  derivation which avoids computing anything explicitly, opting
  instead for a more geometric explanation, which has interesting
  similarities to considerations of Section~\ref{sec:keepingF}.
\end{remark}

  From this we immediately conclude that an isolated non-degenerate
conical point cannot disappear under a perturbation which preserves
the above symmetry.  

\begin{proof}[Proof of Theorem~\ref{thm:persistence} with $\Vbar$ symmetry]
  Surround the point with a small contour $\gamma$, such that inside
  this contour the eigenvalue $\lambda_{-}(\vec{k})$ of
  $H_{\epsilon=0}(\vec{k})$, see (\ref{eq:nondeg_splitting}), is
  simple except at $\vec{k}^*$.  Then on contour $\gamma$ the Berry
  phase of the corresponding eigenfunction must be $\pi$.  

  For small values of $\epsilon$, the eigenvalue on the contour
  $\gamma$ remains simple (as a continuous function on a compact set).
  Therefore, the phase must change continuously, so it must remain
  constant.  Finally, if there were no multiplicity of
  $\lambda_{-,\epsilon}(\vec{k})$ inside the contour, the Berry phase
  would be $0$.  The multiplicity gives rise to a nondegenerate
  conical point by continuity.
\end{proof}

%%%%%%%%%%%%
\subsection{Keeping $F$ symmetry: parity exchange}
\label{sec:keepingF}

Let us now consider the weakly broken $R$ symmetry:
we add to $H$ a perturbation which is $F$-invariant
but not $R$-invariant.
The $V$ symmetry may or may not be preserved.

\begin{proof}[Proof of Theorem~\ref{thm:persistence} with $F$ symmetry]
  As explained in Section~\ref{sec:reduced_symm}, $F$ remains a symmetry
  of the operator $H(\vec{k})$ when the quasi-momenta $\vec{k}$ satisfy
  $\omega_2 = \cc{\omega_1}$ or, equivalently, $k_2 = -k_1$ modulo
  $2\pi$.
  
  Since the subgroup generated by $F$ has two representations, the
  space $\spaceH(\vec{k})$ decomposes into two orthogonal subspaces,
  even and odd, defined by
  \begin{align}
    \label{eq:even_odd_F_def}
    &\spaceH_{F}\sp{+} = \{\psi \in \spaceH(\vec{k}) : F \psi = \psi\} 
    & &\mbox{``even''},\\
    &\spaceH_{F}\sp{-} = \{\psi \in \spaceH(\vec{k}) : F \psi = -\psi\} 
    & &\mbox{``odd''}.
  \end{align}
  All simple eigenvectors of
$H\sb\epsilon(\vec{k})$
  on the symmetry line belong
  to one or the other subspace.  Multiple eigenspaces admit a basis
  consisting of vectors, each of which is either odd or even.
  
  Now suppose we are at the special symmetry point $\vec{k}^*$ in the
  presence of rotational symmetry $R$ (i.e.\ $\epsilon=0$).  At the
  conical point we have a doubly degenerate eigenvalue with orthogonal
  eigenvectors which are mapped into each other by the transformation
  $F$ (see Lemma~\ref{lemma:mult}) and therefore the sum of these
  eigenvectors is even and the difference is odd with respect to $F$.
  
  \begin{figure}[t]
    \centering
    \includegraphics{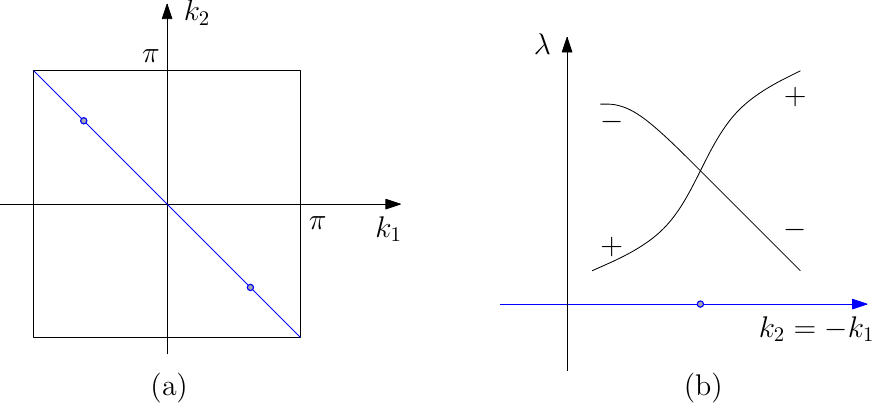}
    \caption{The line in the Brillouin zone where the symmetry $F$ is
      preserved (a).  The form of the dispersion relation along the
      symmetry line (b).}
    \label{fig:flip_symm_line}
  \end{figure}
  
  Now consider the restrictions of the operator
$H\sb\epsilon$ with $\epsilon=0$
  onto the two
  subspaces $\spaceH_{F}\sp{+}$ and $\spaceH_{F}\sp{-}$.  The above
  consideration shows that at the special point each restriction has a
  simple eigenvalue.  As we go along the line $k_2=-k_1$, the
  eigenvalue of each restriction is an analytic function.  These
  functions have an intersection at the point $k_1 =-k_2= 2\pi/3$.
  Since
  the two functions form a section of a non-degenerate cone, the
  intersection is transversal, see Fig.~\ref{fig:flip_symm_line}(b).
  Such intersection is stable under perturbation, and therefore, when
  we consider small $\epsilon \neq 0$ (keeping the symmetry $F$), the
  intersection survives.  Moreover, we know it remains on the line
  $k_2=-k_1$ and the only way it can disappear is by colliding
  with another degenerate eigenvalue on this line.
  
  The intersection
  corresponds to a degenerate eigenvalue of the operator
$H\sb\epsilon(k)$
  which, for small perturbations of the original potential, must still
  be a non-degenerate conical point.
\end{proof}

%%%%%%%%%%%%
\subsection{Destroying all symmetries}
\label{sec:destroy}

When a perturbation breaks all of the symmetries $R$, $V$ and $F$, the
conical point normally separates into two surfaces, locally a
two-sheet hyperboloid.  This was discussed in detail in \cite[Remark
9.2]{FefWei_jams12}.  We merely remark here that the tips of the
sheets of the hyperboloid give rise to the edges of the band
spectrum.  This provides an example for the band edges coming from
a point in the bulk of the Brillouin zone, with no additional
symmetries (since they have been broken), a subject first addressed on
the mathematical level in \cite{HarKucSobWin_jpa07,ExnKucWin_jpa10}.

\appendix
%%%%%%%%%%%%%%%%%%%%%%%%%%%%%%%%%%%%%%%%%%%%%%%%%%%%%%%%%%%%%%%
%%%%%%%%%%%%%%%%%%%%%%%%%%%%%%%%%%%%%%%%%%%%%%%%%%%%%%%%%%%%%%%
%%%%%%%%%%%%%%%%%%%%%%%%%%%%%%%%%%%%%%%%%%%%%%%%%%%%%%%%%%%%%%%

%%%%%%%%%%%%%%%%%%%%%%%%%%%%%%%%%%%%%%%%%%%%%%%%%%%%%%%%%%%%
%%%%%%%%%%%%%%%%%%%%%%%%%%%%%%%%%%%%%%%%%%%%%%%%%%%%%%%%%%%%
\section{Perturbation of pure Laplacian and degeneracy at $\vec{k}=0$}
\label{sec:pert_k0}

In this section we briefly outline the situation at the quasi-momentum
point $\vec{k}=0$ when the operator is $H_0 = -\Delta$.  This should
be compared with the discussion of Section~\ref{sec:pert_pure}.

The lowest eigenvalue of $H_0(0)$ is zero, its only eigenfunction is
the constant function.  The next eigenvalue is six-fold degenerate.
The eigenfunctions are constructed out of the base function
\begin{equation}
  \label{eq:base_eigenfunction0}
  \phi(\vec{x}) := \exp
\big(2\pi i (\vec{b}_1 + \vec{b}_2) \cdot \vec{x}\big)
  = \exp\left(\frac{4\pi i}{\sqrt{3}} x_1\right),
\end{equation}
by rotations.  The symmetries of this problem are the rotation $R$,
inversion $V$, reflection $F$, and complex conjugation $C$.  The group
generated by $R$ and $V$ is the abelian group of rotations by
$2\pi/6$, we denote this rotation by $R_6$.  Then the six orthogonal
eigenvectors are
\begin{equation}
  \label{eq:six_psi}
  \psi_j(\vec{x}) := \sum_{k=0}^6 \sigma^{jk} R_6^k \phi(\vec{x}),
\end{equation}
where $\sigma = \exp(2\pi i /6)$ is the principal 6-th root of
unity.

The six-fold degenerate eigenspace can be decomposed into four
subspaces which correspond to the irreducible representations of the
group of symmetries.  Namely, $\xi =\psi_0(\vec{x})$ satisfies
\begin{equation*}
  R_6\xi = \xi, \qquad F\xi = \xi, \qquad C\xi = \xi,
\end{equation*}
eigenfunctions $\xi = \psi_1(\vec{x})$ and $\eta = -\psi_5(\vec{x})$
satisfy
\begin{equation*}
  R_6
  \begin{pmatrix}
    \xi \\ \eta
  \end{pmatrix}
  =
  \begin{pmatrix}
    \sigma \eta \\ \sigma^5 \xi
  \end{pmatrix},
  \qquad
  F
  \begin{pmatrix}
    \xi \\ \eta
  \end{pmatrix}
  =
  \begin{pmatrix}
    \eta \\ \xi
  \end{pmatrix},
  \qquad
  C
  \begin{pmatrix}
    \xi \\ \eta
  \end{pmatrix}
  =
  \begin{pmatrix}
    \eta \\ \xi
  \end{pmatrix};
\end{equation*}
eigenfunctions $\xi = \psi_2(\vec{x})$ and $\eta = \psi_4(\vec{x})$
satisfy
\begin{equation*}
  R_6
  \begin{pmatrix}
    \xi \\ \eta
  \end{pmatrix}
  =
  \begin{pmatrix}
    \sigma^2 \eta \\ \sigma^4 \xi
  \end{pmatrix},
  \qquad
  F
  \begin{pmatrix}
    \xi \\ \eta
  \end{pmatrix}
  =
  \begin{pmatrix}
    \eta \\ \xi
  \end{pmatrix},
  \qquad
  C
  \begin{pmatrix}
    \xi \\ \eta
  \end{pmatrix}
  =
  \begin{pmatrix}
    \eta \\ \xi
  \end{pmatrix};
\end{equation*}
finally, $\xi = i \psi_3(\vec{x})$ satisfies
\begin{equation*}
  R_6\xi = -\xi, \qquad F\xi = -\xi, \qquad C\xi = \xi.
\end{equation*}

Perturbing the operator $H_0$ by a weak potential
$\varepsilon q(\vec{x})$
which has all the symmetries $\{R,\,V,\,F,\,C\}$ will split
this group of 6 eigenvalues into 4 groups corresponding to the above
representations.

%%%%%%%%%%%%%%%%%%%%%%%%%%%%%%%%%%%%%%%%%%%%%%%%%%%%%%%%%%%%
%%%%%%%%%%%%%%%%%%%%%%%%%%%%%%%%%%%%%%%%%%%%%%%%%%%%%%%%%%%%
\section{Perturbation around a degenerate point with $F$ symmetry}
\label{sec:F_perturbation}

It is interesting to calculate the matrices $h_1, h_2$ if the
degenerate eigenspace has $F$ symmetry.  Suppose the basis is chosen
such that
\begin{equation*}
  F f_1 = f_1, \qquad F f_2 = -f_2.  
\end{equation*}
This can be done at the special point $K$ if the operator has $R$
symmetry; in section~\ref{sec:keepingF} we showed that this
situation survives even if we weakly break the symmetry $R$.

In this case, Lemma~\ref{lem:intertwiner} yields
\begin{equation}
  \label{eq:deriv_matrix_conj_flip}
  \begin{pmatrix}
    1 & 0 \\ 0 & -1
  \end{pmatrix}
  h_{\vec{k}} 
  \begin{pmatrix}
    1 & 0 \\ 0 & -1
  \end{pmatrix}
  = h_{\hat{F}\vec{k}},
  \qquad
  \hat{F} = 
    \begin{pmatrix}
    0 & -1\\
    -1 & 0
  \end{pmatrix}.
\end{equation}

It is easiest to evaluate $h_{\vec{k}}$ in the direction $\vec{k}_e =
(1, -1)^T$, which is an eigenvector of $\hat{F}$ with eigenvalue 1,
and in the direction $\vec{k}_o = (1,1)^T$, which is an eigenvector of
$\hat{F}$ with eigenvalue $-1$.  Remembering that $h_{-\vec{k}} =
-h_{\vec{k}}$, we get
\begin{equation}
  \label{eq:h_flip_ans}
  h_{\vec{k}_e} = 
  \begin{pmatrix}
    a & 0 \\ 0 & c
  \end{pmatrix},
  \qquad
  h_{\vec{k}_o} = 
  \begin{pmatrix}
    0 & b \\ \cc{b} & 0
  \end{pmatrix},
\end{equation}
In particular, the trace of the derivative matrix in the direction
perpendicular to the symmetry line $k_2=-k_1$ is zero and thus the
cone can only be tilted in the direction of the symmetry line.  If $R$
symmetry is present, there is no tilt, as mentioned above.

% The above picture, where the derivative matrix is diagonal in the
% direction preserving a symmetry, and has zeros on the diagonal in the
% orthogonal direction, can be generalized to other space symmetries.

% \begin{proposition}
%   Let $H(\vec{k})$ be an operator depending on quasi-momenta
%   $\vec{k}$.  Assume that at point $K$ it has a unitary $\C$-linear
%   symmetry $S$ whose eigenspaces we denote $X_j$, $j=1,\ldots$.
%   Consider the basis of the dual $\vec{k}$-space, which consists of
%   the eigenvectors of $\hat{S} = S^*$.  If the vector $\vec{e}$ corresponds to
%   eigenvalue 1 then
%   \begin{equation*}
%     \partial_{\vec{e}} H X_j \subset X_j, \quad \mbox{for all }j.
%   \end{equation*}
%   If the vector $\vec{k}$ corresponds to eigenvalue other than 1 then
%   \begin{equation*}
%     \partial_{\vec{e}} H X_j  \perp X_j, \quad \mbox{for all }j.
%   \end{equation*}
% \end{proposition}

%%%%%%%%%%%%%%%%%%%%%%%%%%%%%%%%%%%%%%%%%%%%%%%%%%%%%%%%%
%%%%%%%%%%%%%%%%%%%%%%%%%%%%%%%%%%%%%%%%%%%%%%%%%%%%%%%%%
\section{Berry phase around a conical point}
\label{sec:berry}

Here, for completeness, we give a proof of the fact that the Berry
phase around a nondegenerate conical point is $\pi$, which has been
formulated as Lemma~\ref{lem:minus_one}.  The proof is geometrical in
nature and avoids the direct computation used in the original articles
\cite{HerLon_dfs63,Ber_prsla84}.

Presence of the antiunitary symmetry $\Vbar$ which squares to $-1$
allows us to choose special bases for eigenspaces.  We will be using the
following lemma.

\begin{lemma}
  \label{lem:special_phase}
  Let $A$ be an antiunitary involution on a separable Hilbert space
  $X$.  Then
  \begin{enumerate}
  \item there is an orthonormal basis $\{f_j\}$ of vectors such that
    \begin{equation}
      \label{eq:phase1}
      Af_j = f_j.  
    \end{equation}
     \item if $\dim(X)=2$, there exists a basis $\{\psi, A\psi\}$.
  \end{enumerate}
\end{lemma}

\begin{proof}
  To prove the first part, we start with an arbitrary basis $\{\psi_j\}$.
  Then the vectors
  \begin{equation*}
    f_j^+ = \psi_j + A\psi_j, \qquad \mbox{and} \qquad
    f_j^- = i(\psi_j - A\psi_j)
  \end{equation*}
  both satisfy $Af = f$ and have the vector $\psi_j$ in their span.
  Therefore, the set $\{f_j^+, f_j^-\}$ spans the whole space and can
  be made into a orthonormal basis by applying the Gram-Schmidt process.
  This preserves property \eqref{eq:phase1} since all coefficients
  arising in the process are real:
  \begin{equation*}
\langle f, f'\rangle
= \langle Af, Af'\rangle
= \cc{\langle f, f'\rangle } \in \R.
  \end{equation*}

  To get the second part from the first we start with the orthonormal
  basis $\{f_1,f_2\}$ satisfying \eqref{eq:phase1} and then take
  \begin{equation*}
    \psi = (f_1 + i f_2)/\sqrt{2}, \qquad 
    A\psi = (f_1 - i f_2)/\sqrt{2},
  \end{equation*}
  which can be checked to be orthonormal.
\end{proof}

Now we are in the position to prove Lemma~\ref{lem:minus_one}.

\begin{proof}[Proof of Lemma~\ref{lem:minus_one}]
  Representing the parameters around the location of the conical point
  in polar form we will study the limiting eigenvectors
  \begin{equation}
    \label{eq:limiting_bundle}
    \psi_0^\pm(\theta) = \lim_{r\to0} \psi^\pm(r,\theta),
  \end{equation}
  where $\psi^-$ and $\psi^+$ are the eigenvectors of the lower and
  upper branches of the cone, correspondingly.  We normalize these
  eigenvectors and fix the phase to have 
  \begin{equation}
    \label{eq:phase_fix}
    \Vbar \psi^\pm = \psi^\pm.  
  \end{equation}
  Because the cone
  is nondegenerate (and thus $\left|\lambda_1^+(\theta) -
    \lambda_1^-(\theta)\right|>0$), the limit exists and is continuous
  in $\theta$, see equation \eqref{eq:pert_eigenvectors}.

  \begin{figure}[t]
    \centering
    \includegraphics{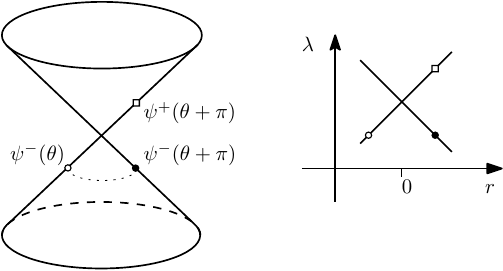}
    \caption{Cone with a schematic representation of a circular
      contour (left); a cross-section of the cone by a plane through
      $\lambda$ axis in the direction $\phi$ (right).}
    \label{fig:line_bundle}
  \end{figure}

  The functions $\psi_0^\pm(\theta)$ have a curious property: since
  the section of the cone by a vertical plane is two intersecting
  lines, Fig.~\ref{fig:line_bundle}, the vector $\psi_0^+(\theta+\pi)$
  is the same as $s_1 \psi_0^-(\theta)$, where $s_1 = \pm1$.

  We expand $\psi_0^\pm$ in a fixed basis of eigenvectors at the
  conical point, which we can choose to be of the form  $\{\phi, \Vbar\phi\}$,
  \begin{equation*}
    \psi_0^\pm = \alpha^\pm(\theta) \phi + \beta^\pm(\theta) \Vbar\phi.
  \end{equation*}
  From condition \eqref{eq:phase_fix} we immediately get $\beta^\pm =
  \cc{\alpha^\pm}$.  On the other hand, the vectors $\psi_0^+$ and
  $\psi_0^-$ are orthogonal, leading to the condition
  \begin{equation*}
    \cc{\alpha^+} \alpha^- + \alpha^+ \cc{\alpha^-} = 0
    \qquad \mbox{or} \qquad
    \cc{\alpha^+} \alpha^- \in i\R.
  \end{equation*}
  From normalization of $\psi_0^\pm$, we conclude that $\alpha^- = i
  \alpha^+ s_2$, where $s_2=\pm1$.  We therefore get
  \begin{equation*}
    \alpha^+(\theta+\pi) = \alpha^-(\theta) s_1 
    = i \alpha^+(\theta)s_1 s_2,
  \end{equation*}
  and, therefore,
  \begin{equation*}
    \alpha^+(\theta+2\pi) = (i s_1 s_2)^2 \alpha^+(\theta)= - \alpha^+(\theta).
  \end{equation*}
\end{proof}

We remark that in the proof above, the overall sign $s_1
s_2$ determines the direction of rotation of the vectors
$\psi_0^\pm(\theta)$ in the two-dimensional space.

\section*{Acknowledgment}

We would like to thank Peter Kuchment for introducing us to the
remarkable paper \cite{FefWei_jams12} which was the starting point for
our exploration.  Rami Band, Ngoc Do, Peter Kuchment and Alim
Sukhtayev patiently listened to our sometimes confused explanations
and provided encouragement, deep suggestions and corrections.  We are
grateful to Yves Colin de Verdi\`ere for his interest in the project
and many helpful suggestions.  Chris Joyner helped us interpret
``strange'' representations
\eqref{eq:RV_representations1}-\eqref{eq:RV_representations2} as
corepresentations of Wigner and gave us a crash course on classifying
them.  Charles Fefferman and Michael Weinstein drew our attention to
several omissions and pointed out the finer points of their results.
We are deeply thankful to all the above individuals.  GB was partially
supported by NSF grant DMS-1410657.
%% PLEASE DO NOT MODIFY THE FOLLOWING CRAZY AWKWARD STATEMENT:
The research of Andrew Comech was carried out
at the Institute for Information Transmission Problems
of the Russian Academy of Sciences
at the expense of the Russian Foundation
for Sciences (project 14-50-00150).

\bibliographystyle{abbrv}
\bibliography{graphene}

\end{document}